\documentclass[12pt]{article}
\usepackage{amssymb}
\usepackage{amsmath}
\usepackage{mathrsfs}
\usepackage{amsthm}
\usepackage[latin1]{inputenc}
\usepackage{txfonts}
\usepackage{hyperref}
\usepackage{color}
\usepackage{bbm}
\usepackage{verbatim} %for using the \begin{comment} and \end{comment}
\usepackage[toc,page]{appendix}
\usepackage[pdftex]{graphicx}
\usepackage{float}
\usepackage{capt-of} %For using \captionof{figure}{A non floating figure}
\usepackage{multirow}
\usepackage{color}
\RequirePackage{natbib}

\newtheorem{thm}{Theorem}[section]
\newtheorem{prop}[thm]{Proposition}
\newtheorem{lemma}[thm]{Lemma}

\theoremstyle{remark}
\newtheorem{rem}[thm]{Remark}

\theoremstyle{definition}

\makeatletter
\newcommand*{\textlabel}[2]{%
  \edef\@currentlabel{#1}% Set target label
  \phantomsection% Correct hyper reference link
  #1\label{#2}% Print and store label
}
\makeatother

\newcommand{\indicator}[1]{\mathbbm{1}_{\left\{ {#1} \right\} }}

\newcommand{\mut}{\mu+\theta}
\newcommand{\less}{<}
\newcommand{\more}{>}
\newcommand{\N}{\mathbb{N}}
\newcommand{\Px}{\mathbb{P}}
\newcommand{\Ex}{\mathbb{E}}
\newcommand{\R}{\mathbb{R}}
\newcommand{\calF}{\mathcal{F}}

\newcommand{\sbt}{\,\begin{picture}(-1,1)(-1,-2)\circle*{2}\end{picture}\ }  %para hacer algo entre cdot y bullet
\newcommand{\oz}{\bar{z}}
\newcommand{\ov}{\tilde{v}}

\newcommand{\ax}{\bar{a}}
\newcommand{\up}{\upsilon}
\newcommand{\ox}{\bar{x}}
\newcommand{\oy}{\bar{y}}

  \definecolor{Reddd}{rgb}{1.00, 0.00, 0.00}

  \definecolor{Red}{rgb}{1.00, 0.00, 0.00}%{1.00, 0.00, 0.00}
    
    \definecolor{DRed}{rgb}{1,0,0}%{0.7, 0.3, 0.00}
    
    \definecolor{Green}{rgb}{1,0,0}%%{0.2, 0.5, 0.2}%{0.5, 0.00, 0.00}
    
    \definecolor{Blue}{rgb}{1,0,0}%{0,0,1}%{0,0,1}%{0.00, 0.00, 1.00}%{0.00, 0.00, 1.00}
    
    \definecolor{PG}{rgb}{1,0,0}%%{.6, .6, .6}
    
    \definecolor{Orange}{rgb}{1,0,0}%{.95, 0.5, 0}
    
        \definecolor{Orangee}{rgb}{1,0,0}%{0.95,0.5,0}%{.95, 0.5, 0}

    \definecolor{Gold}{rgb}{1,0,0}%{0,0,0}%{.864, .498, .196}

		\definecolor{Brilliantrose}{rgb}{1,0,0}%{1.0, 0.33, 0.64}
		
		\definecolor{Brilliantrosee}{rgb}{1,0,0}%{1.0,0.33,0.64}

\usepackage{algorithm}% http://ctan.org/pkg/algorithms
\usepackage{algpseudocode}% http://ctan.org/pkg/algorithmicx

\begin{document}
\title{{One-level limit order book models with memory and variable spread}}
\author{
  Jonathan A. Chávez-Casillas\footnote{{Department of Mathematics and Statistics, University of Calgary, Calgary, AB T2N4C2, Canada ({\tt jonathan.chavezcasil@ucalgary.ca}).}} \\
  %\texttt{first1.last1@xxxxx.com}
  \and
  Jos\'e E. Figueroa-L\'opez\footnote{{Department of Mathematics, Washington University in St.~Louis, MO 63130, USA ({\tt figueroa@math.wustl.edu}).}} \\
  %\texttt{first2.last2@xxxxx.com}
}
%\\[.7cm] Purdue University,\\[.2cm] Department of Mathematics}
\maketitle

\begin{abstract}
We propose {a new model} for the level I of a Limit Order Book (LOB), which incorporates the information about the standing orders at the opposite side of the book after each price change and {the} arrivals of new orders within the spread. 
{Our main result gives a diffusion approximation for the mid-price process. To illustrate the applicability of the considered framework, we also propose a feasible method to compute several quantities of interest}, such as the distribution of the time span between price changes and the probability of consecutive price increments conditioned on the current state of the book. The proposed method is used to {develop} an efficient simulation scheme for the price dynamics, which is then applied to assess numerically the accuracy of the diffusion approximation.

%\vspace{0.2 cm}
%\noindent{\textbf{AMS 2000 subject classifications}: 60G51, 60F99, 91G20, 91G60.}

\vspace{0.2 cm}
\noindent{\textbf{Keywords and phrases}: Limit Order Book Modeling, Price Process Formation, {Heavy Traffic/Diffusion Approximation}.}

\end{abstract}

\section{Introduction}

{The evolution of trading markets has evolved considerably during the last few years.  {Most} modern exchanges use completely automated platforms called Electronic Communication Networks (ECN), which has significantly increased the speed of trading to only a few milliseconds. {ECN} are based on a continuous double auction trading mechanism, in which any trader can submit orders to buy or sell an asset. Two type of orders are available: limit and market orders. A bid (ask) limit order specifies the quantity and the price at which a trader is willing to buy (sell) the asset in question. The so-called Limit Order Book (LOB) aggregates all the outstanding limit orders at any given time. Limit orders with the same price are ranked in a first-in-first-out (FIFO) priority execution. More specifically, a LOB can be visualized as a system of (possibly empty) FIFO queues (one for each possible tick price). Figure \ref{fig:LOBRepresent} gives a graphical representation of a LOB.  The {separation} between the lowest price at the ask side (called {the} ask price) and the {highest} price at the bid side (called {the} bid price) is denominated as  the spread of the LOB, while the queues corresponding to these best bid and ask prices are called the level I of the book.

\begin{figure}[h]
	\centering
		\includegraphics[height=7cm]{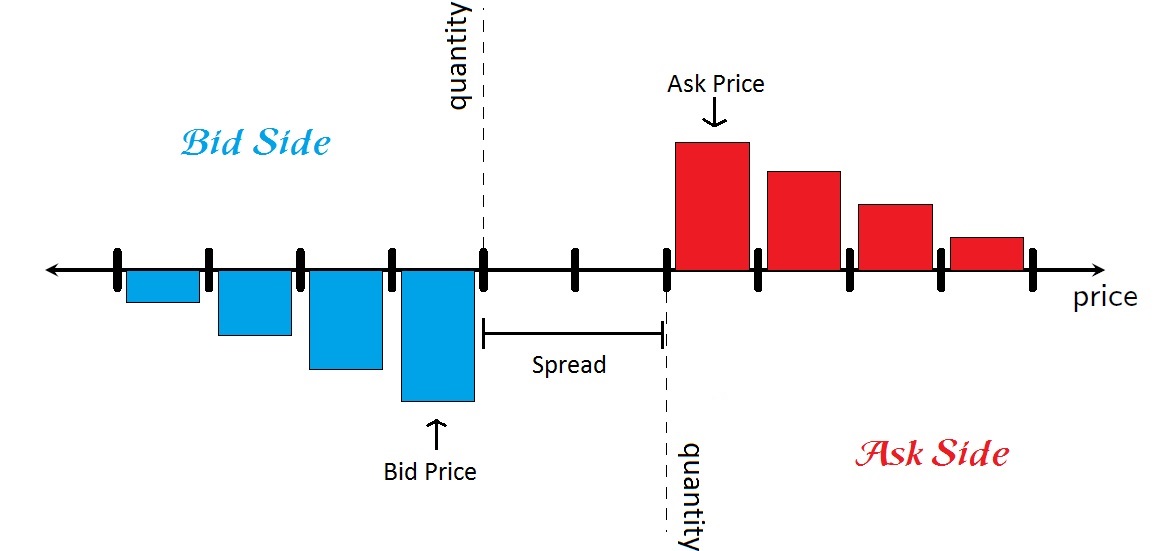}
	\caption{Graphical representation of a Limit Order Book. The Bid Limit orders (to the left) are displayed in blue, while the Ask limit orders (to the right) are displayed in red.}
	\label{fig:LOBRepresent}
\end{figure}

By contrast, market orders are requests to buy or sell a certain quantity of the asset at the best available price. Hence, a market order to sell (buy) the asset is matched against outstanding limit orders sitting at the best bid (ask) queue. 
Other than limit orders and market orders, cancellation of limit orders is another common operation, 
{which account for a considerably large fraction of the operations in an order book (cf. \cite{Harris})}. There is obviously a tradeoff between market and limit orders. While the former are immediately and surely executed, they provide the worst possible price. In contrast, executed limit orders would give better prices but there is a possibility that they won't be executed at all.

From the previous discussion it is clear that a price change occurs when either a queue of  the book's level I gets depleted (due to the cancellation of  limit orders or the arrival of a market order) or a new limit order is posted within the spread. Hence, the best bid (ask) price decreases (increases) when the queue of limit orders at the best bid (ask) gets depleted, in which case the next nonempty bid (ask) queue becomes part of the new level I. Similarly, when the spread is larger than one tick, the best bid (ask) price increases (decreases) when a new bid (ask) order is posted within the spread. It is important to point out that the dynamics of the asset's mid price is determined by the evolution of the book's level I.}

In recent years, there has been a great deal of {attention} on modeling the dynamics of LOBs. {We refer to \cite{Gould} for a recent review on the subject}. 
Earlier works in LOB modeling include \cite{Luckock}, \cite{Kruk}, \cite{Mendelson}, {and \cite{Rosu}} %, and \cite{Kruknew} 
to mention just a few.
More recently, \cite{Cont10} {proposed} a {continuous-time} Markovian model for the order book, in which the possible ask and bid prices of limit orders are assumed to lie in a finite set $\mathcal{C}$. The arrival of limit orders at the different {price levels} in $\mathcal{C}$ are driven by independent Poisson processes, {whose} arrival rates {are} determined by the distance of {their} corresponding price level to the best bid/ask prices {according to} a power law relation. 
{This approach {was} further developed in \cite{Abergel} for a ``finite frame LOB model". More specifically, a fixed number $K$ of ``potential" bid and ask price levels are considered. A potential bid (ask) price level is one at which a bid (ask) limit order could be submitted. In particular, each time the best bid (ask) price queue gets depleted, the frame of potential ask (bid) price levels are shifted to the left (right), hence, forgetting the information at {the} right (left) most level in the other side of the book. Similarly, each time a bid (ask) limit order is submitted within the spread, the frame of possible ask (bid) price levels is {shifted} to the right (left), for which queues at the right (left) most price levels are assumed constant. Under the just described ``fixed moving frame" assumption, a diffusive limit for the mid-price process is established therein.  See below for further discussion and comparisons with our results.}

Our main {inspiration} for the present work is drawn from \cite{CL2012}'s seminal work, where a Markovian model is considered for the dynamics of level I. {The main motivation of considering only the level I and not the entire book is twofold. Firstly, as pointed {out} above, the asset's price is determined by the level I and, secondly, the information contained in the level I is key for many high-frequency trading strategies and problems.} By {preserving a} Poissonian order flow and imposing {some symmetry conditions on the shape} of the order book, \cite{CL2012} prove that the mid-price process, properly scaled in time and space, converges to a Brownian motion. The volatility of the limiting process can explicitly be computed using the input parameters of the model. {Unfortunately, the results therein required several strong assumptions, the most important of which are:}
%\vspace{-.1 cm}
\begin{itemize}
	\item[(i)] {a constant volume for all type of orders: market, limit, and cancellations};
	\item[(ii)] a constant spread {of one tick} between the best ask and bid prices {at all times};
	\item[(iii)] constant {parallel} price shifts {of one tick} after each depletion of a level I queue;
	\item[(iv)] {loss of memory}, in the sense that, after {each} level I queue depletion, the information on the remaining limit orders at the side which was not depleted is {reset}.
\end{itemize}

The assumption of constant spread is generally well-justified for ``large" tick assets, in which the tick size is comparable to a ``typical" price increment. {Indeed, as illustrated in \cite{CL2012}, for some stocks in the US market, the spread could be equal to one tick for more $98\%$ of the observations within a typical day. Besides small-tick assets (e.g.,  EUR/USD FX), there are other situations in which} the constant spread assumption cannot be validated. For instance, \cite{Bouchaud} argues that {relatively large spreads may be created by} monopolistic practices of market makers, high order processing costs, and large market orders. 
\cite{Pomponio} (see also \cite{MuniToke2012}) {argue} that traders keep track of the amounts of orders at the best quote in the {LOB and typically restrict the size of their market orders to be less than these amounts}. But, sometimes the speed of execution is more important than the market impact risk of large orders. {In that} case, orders larger than the size of the first limit (called trades-through) may be submitted. \cite{Pomponio} {argue} that {even though} trade-throughs may rarely occur (with {an} {occurrence probability} of less than $5\%$), they {make up for} a non-negligible part of the daily-volume (up to $20\%$ for the DAX index future). Since every trade-through widens the spread, a model that allows for variable spread is desirable.

In this work, {we propose a new model as a way to account for the possibility of a variable spread,  as described in the previous paragraph, and to relax some of the shortcomings of the framework proposed in \cite{CL2012}.
More concretely, while keeping some of the assumptions therein,} such as the Markovian order flow, one level at each side of the book, and a constant volume of order flow,
{the proposed} model allows for some ``memory" in the dynamics of the LOB by keeping the information of the outstanding orders on the other side of the book after each level I queue depletion. For instance, if the best bid queue gets depleted, the best bid price decreases one tick, but both the {price} and outstanding orders of the best ask price {are preserved}. In order to avoid perpetual widening of the spread, we also allow the arrival of orders within the spread according to a Poisson process. 
{As in \cite{CL2012}, we establish a diffusive approximation for the price process, albeit using an essentially different analysis. Concretely, our results in this direction build} on the mathematical {theory} of countable positive recurrent Harris chains (see, e.g., \cite{Meyn}).

One of the main {appealing} features of the model in \cite{CL2012} is its {tractability, which, in particular, enables} analytical computation of {several LOB features}, some of which are relevant for {high-frequency} trading and intraday risk management. {To illustrate the applicability of the framework proposed in this work, we put forward a feasible method to compute several quantities of interest, conditioned on the initial state of the book,} such as the distribution of the duration between price changes, the probability of a price increase, and the probability of two consecutive increments on the price. {The main tool for the derived formulas is an explicit characterization of the joint distribution of the time of a {depletion at the level I} and the amount {of} orders at the remaining queue at such {a} time based on the eigenvalues and eigenvectors of a certain finite-difference operator. The developed methods are also applied to {devise}} an efficient simulation algorithm for the dynamics of the LOB, 
{which is subsequently used to numerically study} the convergence of the midprice process towards its diffusive limit process.
The results in this manuscript are on one hand {theoretically relevant as {they allow to relax some of the strong assumptions of earlier works. On the other hand,} these results {enable} to give a bottom-up construction of the some ``efficient" models commonly used at low frequencies (e.g., subordinated Brownian motion models). 

{Let us finish this introduction with a brief discussion about the connection of our work with some earlier works.}
As previously mentioned, {\cite{Abergel} also obtains a diffusive approximation for the mid-price process;} however, our model cannot be framed within the approach {therein}. To {realize} this, note that the spread {the referred paper} remains bounded, while in our model, it can potentially take values on {$\mathbb{Z}_{+}$}, which considerably complicates the analysis. Let us also remark that our model can still account for the empirical {observation} of \cite{CL2012} that the spread spends a {very} large amount of time at the value of $1$ by taking {a large value for the intensity of {limit order arrivals} within the spread}. This feature {is, however, not possible} to incorporate {in the model in \cite{Abergel}} since the intensity of arrivals at the first potential bid and ask price level is constant regardless this level is already occupied by limit orders or not.

{Another relevant work is \cite{CL2012b}, which relaxes the assumptions of constant volume and no memory of \cite{CL2012} (assumptions (i) and (iv) as described above), but does not establish a diffusive approximation for the price process. Instead, the main result therein is a heavy traffic approximation for {the queue sizes of the LOB level I} as a Markovian jump-diffusion process in the first {quadrant}. On the interior of the {first quadrant}, the process follows a planar Brownian motion with some given drift and covariance matrix.  {Every time that the process hits one of the axis, this is then shifted to} the interior of the {quadrant} (according to some rules). {Therefore, as described in Proposition 1 in \cite{CL2012}, the price process is approximated by a random walk which moves one tick to the left or right depending on whether the just described jump-diffusion process hits} the x- or y-axis. Note that this is not the same type of {continuous} diffusive approximation for the price process {as} considered in this work. Let us also remark that, even though, in principle, \cite{CL2012b} relaxes the assumption (iv) of no memory, it imposes other technical assumptions that apparently precludes the type of memory considered in this work (see Remark \ref{RCL2012Ex} below for further details).}

The remainder of this paper is organized as follows. Section \ref{MainModel} analyses the model with variable spread and  
obtains a {Functional Central Limit Theorem} for the resultant price process. {Section \ref{CmpImpQnt} introduces {the method described above} to efficiently compute several quantities of interest related to the LOB.} In Section \ref{Implementation}, we analyze, via simulations,} the rate of convergence to the limiting process and the behavior of the spread and the asymptotic volatility in relation to the different model's parameters. To this end, we develop an ``efficient'' method to simulate the dynamics of the LOB level I based on our results of Section \ref{CmpImpQnt}. {In Section \ref{Implementation}, we also compute numerically some of the quantities of interest considered in Section \ref{CmpImpQnt} and study their behaviors under both our assumptions and those in \cite{CL2012}.} Finally, some of the technical proofs are presented in {one appendix}.

\section{A one-level LOB model with {memory and variable spread}} \label{MainModel}

In this section, we introduce {a new framework for the dynamics of the level I of a LOB}, which {is more} realistic than the model introduced in {\cite{CL2012}} in that, whenever a queue is depleted,  {the position and {size} of the remaining queue {are preserved}, while a new queue is generated one tick to the left or right of the depleted queue depending on whether the latter was at the bid or ask side of the book.} The model also incorporates {the possibility of {limit order} arrivals} within the spread whenever this is possible. {The main result in this section is to establish a diffusive approximation for} the mid-price process. More specifically, if $\{s_{t}\}_{t\geq{}0}$ {denotes} the {mid-price} process of the stock, then, for some appropriate constants $\sigma>0$ and $m$, the following invariance principle holds: 
\begin{equation}\label{GlAsMp}
	{\dfrac{s_{nt}-nmt}{\sqrt{n}} \Rightarrow \sigma W_t,\qquad n\to\infty},
\end{equation}
where $\{W_{t}\}_{t\geq{}0}$ is a Wiener process {and, hereafter, $\Rightarrow$ denotes convergence in distribution}. Heuristically, if we think of {$1/n$ as the time scale} at which the process is observed, (\ref{GlAsMp}) says that the price process can be approximated by a Brownian motion with drift at {``small" scales} (typically, 10 or more seconds, depending on the speed of the book events that happen at the order of milliseconds).

{This section is organized as follows. We first introduce the model and necessary notation in Subsection \ref{Sec:Mod2}. Subsection \ref{LLNmodtau} proves a Law of Large Numbers for the inter arrival times between price changes, which in turn is needed to determine the appropriate time scaling of the price process.  Finally, we proceed to obtain a Functional Central Limit Theorem (FCLT) for the price process itself in Subsection \ref{longrunsec}.}

\subsection{LOB dynamics}\label{Sec:Mod2}

As previously explained, {we only consider the level I of the order book, which suffices to determine the evolution of the price process of the asset. Concretely, as explained in the introduction, a} price change can only occur when {the outstanding} orders at any side in the book are depleted or when, if possible, a set of new orders arrive within the spread and becomes {part of the new} level I of the book. Concretely, suppose that the best bid price is at {$S^{b}$} and that its queue gets depleted after a market order or cancellation. Then, a new best bid queue is ``generated" at price {$S^{b}-\delta$}, hence causing the spread to widen. The size of this new queue of limit orders is {assumed} to be generated from a distribution ${f^{b}}$ on {$\mathbb{Z}_{+}$}, independently of any other information of the LOB, while the amount and position of the limit orders at the best ask level {are kept} unchanged. Similarly, if the queue at the best ask price gets depleted after a market order or cancellation, then a new queue is generated at the price {$S^{a}+\delta$, where $S^{a}$} is the {best ask} price before the order. The size of the new best ask queue is assumed to be generated from a distribution ${{f}^{a}}$ {on $\mathbb{Z}_{+}$}, independently of any other information. {In that case, the} bid side of the book remains unchanged. 

{The distributions ${f^{a}}$ and ${{f}^{b}}$ are {meant} to reflect the stationary behavior of the queue sizes at the next best queues after depletion.}
Throughout, we assume that both distributions ${f^a}$ and ${{f}^b}$ are supported on $\{1,2,\dots,N^{*}\}$, for some fixed $N^{*}\in{\mathbb{Z}_{+}}$, which can be chosen arbitrarily large. This simplifying assumption is imposed in order to guarantee the recurrence of the underlying Markov chain driving the dynamics of the price process. Also, for simplicity, {the tick is set to be} $\delta=1$.

When the spread is more than 1, {there is also the possibility of a price change due to the arrival of} a new set of orders within the spread at either the ask or bid side of the LOB. In the former case, the best ask price decreases by $\delta$, while the bid side remains unchanged. In the latter case, the best bid price increases by $\delta$, while the other side of the order book does not change. As before, the size of a new queue of limit orders is generated from the distribution ${f^{a}}$ or ${{f}^{b}}$, independently of any other variables, depending on whether the new limit order is at the ask or bid side\footnote{{{Our results are} still valid if one takes these distributions {to be} different from the one used when a level I queue gets depleted}.}

{We now proceed to give a {formal} mathematical formulation of the {LOB} dynamics. To that end, we need some notation:}
\begin{itemize}
	\item[{(i)}] Let {$\zeta_0$} be the initial spread and {$\zeta_i$}, for $i\geq1$, be the spread after the $i^{th}$ price change. The sizes of the {best} ask and bid queues at time $t$ are denoted by $q_t^a$ and $q_t^b$, respectively. Also, for $i\geq{}1$, {$\tau_i$} represents the time span between the $(i-1)^{th}$ and $i^{th}$ price changes {and we set $\tau_{0}=0$.}
	\item[{(ii)}] Throughout,  $\{\hat{Y}^{a,i}\}_{i\geq0}$ and $\{\hat{Y}^{b,i}\}_{i\geq0}$ are independent sequences of i.i.d. random variables, {taking values on {${\Omega}_{N^*}:=\{1,2,\dots,N^{*}\}$}, and with respective} distributions ${f^{a}}$ and ${f^{b}}$. {These will} indicate the amount of orders at the {best ask or bid queues} after that particular side changes in price.
	\item[{(iii)}] Let $\{{L^{a}_{{i}}(\zeta)}\}_{{i\geq{}0},\zeta\in\mathbb{Z}_{+}}$ and $\{{L_{{i}}^{b}(\zeta)}\}_{{i\geq{}0},\zeta\in\mathbb{Z}_{+}}$ be independent sequences of independent random variables such that ${L^{a}_{{i}}(\zeta)}$ and ${L^{b}_{{i}}(\zeta)}$ are exponentially distributed with parameter $\alpha\indicator{\zeta>1}$\footnote{{In particular, $L_{{i}}(1)=M_{{i}}(1)=\infty$, a.s., for all ${i}$.}}. These variables are also independent of any other variables in the system. Hereafter, ${L_{{i}}}(\zeta):={L_{{i}}^{a}(\zeta)\wedge L^{b}_{{i}}(\zeta)}$. {We shall interpret ${L_{i}^{a}(\zeta)}$ and ${L^{b}_{i}(\zeta)}$ as the times for a new set of orders to arrive at the ask and bid side, respectively, after the $i^{th}$-price change when the spread is at the value $\zeta$.}  
	\item[{(iv)}] For any starting point {$x\in\bar{\Omega}_{N^*}:=\{0,1,\ldots,N^*\}$}, let {$Q(x):=\{Q_t(x)\}_{t\geq{}0}$} be {a} continuous time Markov process with state space $\bar{\Omega}_{N^*}$ such that {$Q_0(x)=x$} and {its transition matrix ${\mathcal{Q}}:{\bar{\Omega}_{N^*}}\times {\bar{\Omega}_{N^*}}\to \R$ is given} by:
  \begin{align}\label{TrnsMtrQ} \nonumber
		{\mathcal{Q}}_{{j,j+1}}=\lambda,\;\;\text{for }{0}\leq {j}\leq N^*-1,&\quad 
		{\mathcal{Q}}_{{j,j-1}}={\upsilon}, \;\; \text{for } 1\leq {j} \leq N^*,\\
		{\mathcal{Q}}_{{j,j}}=-({\upsilon}+\lambda), \;\; \text{for } 1\leq {j} \leq N^*-1,&\quad \mathcal{Q}_{N^*,N^*}={-\upsilon},
		\\ \nonumber
		{\mathcal{Q}}_{{j,\ell}}=0,& \;\; \text{{otherwise}}
	\end{align}
{where $\upsilon:=\mu+\theta$ and $\lambda,\mu,\theta\in(0,\infty)$ are interpreted as the intensity of arrivals of limit orders at the level I, market orders, and cancellations, respectively.}	
\item[{(v)}] Finally, for any $i\geq{}0$ and $x\in\bar\Omega_{N^{*}}$, we let $Q^{a,i}(x):=\{Q^{a,i}_{t}(x)\}_{t\geq{}0}$ and $Q^{b,i}(x):=\{{Q}^{b,i}_{t}(x)\}_{t\geq{}0}$ be processes such that 
\begin{equation}\label{DfnQaQb}
	{Q^{a,i}(x)\stackrel{\mathcal{D}}{=}Q^{b,i}(x)\stackrel{\mathcal{D}}{=}Q(x)},
\end{equation} 
and  the collection of processes $\{Q^{a,i}(x),Q^{b,i}(x)\}_{i\geq{}0,x\in\bar\Omega_{N^{*}}}$ are mutually independent, and also independent of the processes introduced in {the points (ii)-(ii)}.
\end{itemize}

We are ready to give a formal construction {of the LOB dynamics}. Fix {$\tau_0=0$ and define the processes 
\begin{equation}\label{DfnXYOrg}
	X_t^{a,0}:=Q^{{a,0}}_t(x_{0}^{a}),\quad\quad
	X_t^{b,0}:=Q^{{b,0}}_t(x_{0}^{b}),
\end{equation}
for some arbitrary random initial queue sizes {$(x_{0}^{a},x_{0}^{b})\in{\Omega}_{N^*}^{2}$}, {which are assumed to be independent of any of the other processes considered in the points (i)-(v) above}. 
With the notation
\[
	\sigma^{a,1}:=\inf\left\{t\geq0:X_t^{a,0}=0\right\}\wedge {L^{a}_0(\zeta_{0})},\qquad 
	\sigma^{b,1}:=\inf\left\{t\geq0:X_t^{b,0}=0\right\}\wedge {L^{b}_0(\zeta_{0})},
\] 
at hand, the time of the first price change can now be defined by
\begin{equation}\label{FCP}
	T_{1}:={\tau_1}:=\sigma^{a,1}\wedge\sigma^{b,1},
\end{equation}
while, {for $t\in[0,T_{1})$,} the queue sizes at the best ask and bid prices are {respectively} given by
\[
	q_t^a=X_t^{a,0},\qquad q_t^b=X_t^{b,0}.
\]
{The number of orders at each side of the LOB {and the spread at time $T_1$ are} then set as}
\begin{align*}
	&q_{T_{1}}^{a}:=x_{1}^{a}:=\hat{Y}^{a,1}\indicator{{\tau_1}=\sigma^{a,1}}+X_{{\tau_1}}^{a,0}\indicator{{\tau_1}=\sigma^{b,1}},\quad
	{q_{T_{1}}^{b}}:=x_{1}^{b}:=\hat{Y}^{b,1}\indicator{{\tau_1}=\sigma^{b,1}}+X_{{\tau_1}}^{b,0}\indicator{{\tau_1}=\sigma^{a,1}},\\
	&\zeta_{1}=\zeta_{0}+\indicator{\tau_{1}<{L_0^{a}(\zeta_{0})\wedge L_0^{b}(\zeta_{0})}}-\indicator{\tau_{1}={L^{a}_0(\zeta_{0})}}-
	\indicator{\tau_{1}={L_0^{b}(\zeta_{0})}}.
\end{align*}
This process is continued recursively. Concretely, for $i\geq{}1$, we set
\begin{align*} 
	&q_t^a:=X_{t-T_{i}}^{a,i}, \quad q_t^b:=X_{{t-T_{i}}}^{b,i}, \quad \text{for } t\in\left[{T_i,T_{i+1}}\right),\quad {T_{i+1}:=T_{i}+\tau_{i+1}},\\
	&\tau_{i+1}:=\sigma^{a,i+1}\wedge\sigma^{b,i+1},\quad \zeta_{i+1}=\zeta_{i}+ \indicator{\tau_{i+1}<{L_i^{a}(\zeta_{i})\wedge L_i^{b}(\zeta_{i})}}-\indicator{\tau_{i}={L_i^{a}(\zeta_{i})}}-
	\indicator{\tau_{i}={L_i^{b}(\zeta_{i})}}
\end{align*}
where
\begin{align*} 
	&	X_t^{a,i}=Q^{{a,i}}_{t}(x_{i}^{a}),\qquad	X_t^{b,{i}}=Q^{{b,i}}_{t}(x^{b}_{i}),\\
		&\sigma^{a,i+1}=\inf\left\{t\more0:X_t^{a,i}=0\right\}\wedge {L_i^{a}(\zeta_{i})}, \qquad 
	\sigma^{b,i+1}=\inf\left\{t\more0:X_t^{b,i}=0\right\}\wedge {L_i^{b}(\zeta_{i})},\\
	&x^{a}_{i+1}:=\hat{Y}^{a,i+1}\indicator{\tau_{i+1}=\sigma^{a,i+1}}+ X_{\tau_{i+1}}^{a,i}\indicator{\tau_{i+1}=\sigma^{b,i+1}},\qquad
x_{i+1}^{b}:=\hat{Y}^{b,i+1}\indicator{\tau_{i+1}=\sigma^{b,i+1}}+X_{\tau_{i+1}}^{b,i}\indicator{\tau_{i+1}=\sigma^{a,i+1}}.
\end{align*}

The }above formulation justifies the following identities:
\begin{align}\label{MPI}
	&\Px\left(\left.(\tilde{x}_{k},\zeta_{k},\tau_{k})\in B\times C\times D\right|(\tilde{x}_{k-1},\zeta_{k-1})=(\tilde{x},\zeta)\right)=\Px\left(\left.(\tilde{x}_{1},\zeta_{1},\tau_{1})\in B\times C\times D\right|(\tilde{x}_{0},\zeta_{0})=(\tilde{x},\zeta)\right),\\ \label{MPII}
	&\Px\left(\left.(\tilde{x}_{k},\zeta_{k},\tau_{k})\in B\times C\times D\right|{\left\{(\tilde{x}_{i},\zeta_{i},\tau_{i})\right\}_{i=0}^{k-1}}\right)=\Px\left(\left.(\tilde{x}_{k},\zeta_{k},\tau_{k})\in B\times C\times D\right|(\tilde{x}_{k-1},\zeta_{k-1})\right),
\end{align}
where  {$\tau_{0}:=0$ and $\tilde{x}_{k}:=(x_{k}^{b},x_{k}^{a})$ represent's the sizes at the best bid and ask queues after the $k^{th}$ prices change}. In particular, it follows that
\begin{equation}\label{MIRNL}
	\tau_{k} \underset{\{(\tilde{x}_{i-1},\zeta_{i-1})\}_{i\geq{}0}}{\bot} (\tau_{k-1},\dots, \tau_{1}),\quad k\geq{}2,
\end{equation}
which, also implies the mutual independence of $\{\tau_{1},\dots, \tau_{n}\}$ given $\{(\tilde{x}_{i-1},\zeta_{i-1})\}_{i\geq{}0}$. {Furthermore, it is easy to see that the process 
\[
	(\Xi_{t},\Upsilon_{t}):=\sum_{k=0}^{n} (\tilde{x}_{k},\zeta_{k}){\mathbbm{1}_{[T_{k},T_{k+1})}(t)}
\]
is semimarkov in the sense of \cite{Cinlar75}.}

\begin{rem}\label{RCL2012Ex}
	{As mentioned in the introduction, one of the key features of the model proposed above is the incorporation of memory. Other recent works have also considered this feature. Notably, \cite{CL2012b} assumes that the level I queue {sizes after each price change,  $q_{_{T_{i}}}$, is a function  $g(q_{_{T_{i}^{-}}},\varepsilon_{i})$} of the level I queue sizes before the price change, $q_{_{T_{i}^{-}}}$, and a sequence of i.i.d. random innovations $\{\varepsilon_{i}\}_{i\geq{}1}:=\{(\varepsilon^{b}_{i},\varepsilon^{a}_{i})\}_{i\geq{}1}$. One of the examples considered therein is the case of ``pegged limit orders" in which $q_{_{T_{i}}}^{a}=\beta q_{_{T_{i}^{-}}}^{a}+\varepsilon^{a}_{i}$ and $q_{_{T_{i}}}^{b}=\varepsilon^{b}_{i}$, when the best bid queue gets depleted (with a similar relation holding for the case when the best ask queue gets depleted). Here, $\beta$ is a constant proportion and $(\varepsilon^{b}_{i},\varepsilon^{a}_{i})$ have distribution $F$ in {$\mathbb{Z}_{+}^{2}$}. However, Assumption 3 in \cite{CL2012b} precludes the situation where $\varepsilon^{a}_{i}=0$, a.s., and $\beta=1$, which is the type of memory we consider in this work.} 
\end{rem}

\subsection{A law of large numbers for the modified interarrival times }\label{LLNmodtau}

{Our first ingredient toward (\ref{GlAsMp}) is to establish} a law of large numbers (LLN) for {the time of the $n^{th}$-price change, $T_{n}=\sum_{k=1}^{n}\tau_{k}$,} using ergodic results for Markov chains. To that end, we first {introduce} some needed notation. Let $Z=\{Z_t\}_{t\in\N}$ denote a Markov chain {on a probability measure $(\Omega,\mathcal{F},\Px)$} with countable state space $\Xi$ and transition probability matrix ${P}:\Xi\times\Xi\rightarrow[0,1]$. For any probability measure $\mu=\{\mu(\hat{y}),\hat{y}\in\Xi\}$ on $\Xi$, {$\hat{y}\in\Xi$}, and {$A\subset\Xi^{{\mathbb{Z}_{+}}}=\{(z_1,z_2,\ldots)\left.\right|z_i\in\Xi\}$}, {we} define 
\[
{{\Px}_{\hat{y}}(Z_{\sbt}\in A):={\Px}(Z_{\sbt}\in A|Z_0=\hat{y})},\qquad {\Px}_\mu(Z_{\sbt}\in A):=\sum\limits_{\hat{y}\in\Xi}\mu(\hat{y}){\Px}_{\hat{y}}(Z_{\sbt}\in A).
\]
{As usual,} $\Ex_\mu$ denotes the expectation with respect to the probability measure ${\Px}_{\mu}$.
{We} say that {an event} $A$ occurs ${\Px}_*$-a.s. if $A$ occurs ${\Px}_{\hat{y}}$-a.s. for all $\hat{y}\in\Xi$. 

Let us recall that an irreducible Markov chain on a countable state space $\Xi$ is either {transient} or recurrent,
while a set $A$ is called Harris recurrent 
if
\[
{\Px}_{{z}}\left(\sum_{n=1}^\infty {\indicator{Z_n\in A}}=\infty\right)=1,\qquad z\in A.
\]
{A Markov chain is} called Harris recurrent if it is irreducible and every set {$A$ is} Harris recurrent. Also, the Markov chain $Z$ is called positive if it is irreducible and admits an {invariant probability measure}, while a positive and Harris recurrent chain is called positive Harris  (cf. Chapter 10 \cite{Meyn}).
{The following result from \cite{Meyn} (Theorem 17.0.1 therein) is key to obtain the aforementioned LLN.}
\begin{thm}\label{LLNHarris} Suppose that {$\{Z_{{t}}\}_{{t\in\mathbb{N}}}$} is a positive Harris chain with invariant probability measure $\pi$. Then, for any $g$ satisfying $\pi(|g|):=\sum\limits_{x}\pi(x)|g(x)|\less\infty$, 
\[
\lim_{n\rightarrow\infty}\frac{1}{n}\sum_{{t=1}}^n g({Z_{t}})=\pi(g),\quad\quad {\mathbb{P}}_*\hbox{-a.s.}
\]
\end{thm}
In the sequel, we shall use Theorem \ref{LLNHarris} to show a LLN for $T_{n}=\sum_{i=1}^{n}\tau_{i}$ by expressing each $\tau_{i}$ in terms of an appropriate Markov chain $Z:=\{Z_t\}_{t\in\mathbb{N}}$. Concretely, throughout the remaining of this subsection, we take 
\begin{equation}\label{MnZSec3}
	{Z:=\{Z_t\}_{t\geq{}1}:=\left\{(\tilde{x}_{t-1},\zeta_{t-1},\tilde{x}_{t},\zeta_{t})\right\}_{t\geq1}},
\end{equation}
where we recall that { $\tilde{x}_{t}:=(x_{t}^{b},x_{t}^{a})$} and $\zeta_{t}$ respectively represent the number of orders at the {book's level I (bid and ask) and} the spread after the $t-$th price change ({see} Section \ref{Sec:Mod2} {for details about the notation}). By (\ref{MPII}), we can see that {$Z$} is a Markov chain {with countable} state space 
\[
	\Xi:={\left\{(y_1,c_1,y_2,c_2)\ |\ {y_1=(y_1^a,y_1^b)\in{{\Omega}_{N^*}^{2}}},\ y_2=(y_2^a,y_2^b)\in{{\Omega}_{N^*}^{2}},\ c_1,c_2\in{\mathbb{Z}_{+}}, \ |c_1-c_2|=1\right\}}.
\]
Furthermore, {fixing  $U_n:=(\tilde{x}_{n},\zeta_{n})$ and noting that $U:=\{U_n\}_{n\geq{}0}$} is itself a Markov chain {by (\ref{MPI})-(\ref{MPII})}, it follows that
\[
{P(\hat{y},\hat{z})}:=\Px(Z_n=\hat{z}\left.\right| Z_{n-1}=\hat{y})=\Px((U_{n-1},U_n)=\hat{z}\left.\right|(U_{n-2},U_{n-1})=\hat{y})=
\Px(U_n=(z_2,d_2)\left.\right|U_{n-1}=(z_1,d_1)),\]
where $\hat{y}:=(y_1,c_1,y_2,c_2)\in\Xi$ and $\hat{z}:=(z_1,d_1,z_2,d_2)\in\Xi$ with $(y_{2},c_{2})=(z_{1},d_{1})$. 

Our first {objective} is to prove that we can apply Theorem \ref{LLNHarris} to the chain $Z$ {introduced in (\ref{MnZSec3})}. 
\label{ZnHarris} Since, for a countable state Markov chain, irreducibility reduces to see that all states communicate to one another, by {the description of the dynamics of $Z$ given in} the previous section, $Z$ is {clearly} irreducible. 
The existence of the invariant {probability} measure
would hold provided that $Z$ is {positive} recurrent (cf. {\cite[Corollary I.3.6]{Asmussen}}). 
Furthermore, since for a countable-state Markov chain, Harris recurrence is equivalent to {plain} recurrence ({see the} discussion below Theorem 9.0.1 in \cite{Meyn}), $Z$ will {then} be positive Harris chain, {provided that $Z$ is positive recurrent}.
\begin{thm} \label{recurrence} {If {$\alpha\geq{} \mu+\theta$}, then} the Markov chain {$Z:=\{Z_t\}_{t\geq{}1}:=\{(\tilde{x}_{t-1},\zeta_{t-1},\tilde{x}_{t},\zeta_{t})\}_{t\geq1}$} is {positive} recurrent.
\end{thm}

A well-known sufficient condition for a Markov chain to be positive recurrent over a countable state space is given by the following so-called {Foster or} mean drift conditions (cf. Theorem I.5.3 in \cite{Asmussen}) for some function $h:\Xi\rightarrow\R$, a constant $\epsilon> 0$, and a finite set $F\subset\Xi$:
\begin{equation}\label{FosterCnd}
	{\rm (i)}\,{\inf_{{\hat{z}\in \Xi}} h(\hat{z})} >-\infty,\quad {\rm (ii)}\,\sum\limits_{{\hat{z}\in \Xi}}P(\hat{y},\hat{z})h(\hat{z})\less\infty,\; \hat{y}\in F,\quad
	{\rm (iii)}\,\sum\limits_{\hat{z}\in \Xi}P(\hat{y},\hat{z})h(\hat{z})\less h(\hat{y})-\epsilon,\; \hat{y}\notin F.
\end{equation}
In order to verify that $Z$ satisfies the previous conditions, we need two preliminary results.
The following result constructs a super-harmonic function $\varphi$, outside the set ${F_{0}}:=\{\hat{y}\in\Xi: \hat{y}=(y_1,2,y_2,1)\}$. {Recall that $\varphi$ is said to be a super-harmonic function (cf. Section 17.1.2 in \cite{Meyn})  at some $\hat{y}\in\Xi$ if 
\begin{equation}\label{DfnSAF}
{{(P\varphi)(\hat{y}):=\sum\limits_{\hat{z}\in{\Xi}}P(\hat{y},\hat{z})\varphi(\hat{z})\leq\varphi(\hat{y})}.}
\end{equation}
The proof of the next result is deferred to Appendix \ref{ApndB}.}
\begin{lemma} \label{phiharm} Under the notation in Section \ref{Sec:Mod2}, let $\varsigma(x):=\inf\{t\more0:{Q_t^{a,0}(x_{1})\wedge Q_t^{b,0}(x_{2})}=0\}$, for $x:=(x_{1},x_{2})$, and let ${L:=L_{1}^{a}(2)\wedge L_{1}^{b}(2)}$ (i.e., {$L$} is exponentially distributed with rate $2\alpha$). Also, for any $\hat{y}=(y_1,j\pm1,y_2,j)\in\Xi$, let $\varphi(\hat{y}):{\Xi}\rightarrow\R$ be given by 
\begin{equation}\label{phisupharm}
\varphi((y_1,j\pm1,y_2,j)):={{\varphi}(j)}:=\left\{\begin{array}{lll} \left(\frac{1+\sqrt{1-4p_{\mathbf{1}}(1-p_{\mathbf{N^*}})}}{2p_{\mathbf{1}}}\right)^j &\qquad \text{if}& \qquad p_{\mathbf{1}}(1-p_{\mathbf{N^*}})\less\frac{1}{4},\\[.5cm]
\left(\frac{1}{2p_{\mathbf{1}}}\right)^j&\qquad \text{if}&\qquad p_{\mathbf{1}}(1-p_{\mathbf{N^*}})=\frac{1}{4},\\[.5cm]
\left(\frac{1-p_{\mathbf{N^*}}}{p_{{\mathbf{1}}}}\right)^{\frac{j}{2}}\cos(j\theta)&\qquad \text{if}&\qquad p_{\mathbf{1}}(1-p_{\mathbf{N^*}})\more\frac{1}{4},\end{array}\right.
\end{equation}
where 
\[
	p_{\mathbf{1}}:=\Px({L}\more\varsigma((1,1))),\quad p_{\mathbf{N^*}}:=\Px({L}\more\varsigma((N^*,N^*))), \quad \theta:=\arctan(\sqrt{4p_{\mathbf{1}}(1-p_{\mathbf{N^*}})-1}),
\]
Then, $\varphi$ is a super-harmonic function for the process $Z$ given in (\ref{MnZSec3}), at any {$\hat{y}\in\Xi\backslash F_{0}$,  where $F_{0}:=\{\hat{y}\in\Xi: \hat{y}=(y_1,2,y_2,1)\}$.}
\end{lemma}
The next {result} is crucial to construct the function $h$ satisfying the conditions (\ref{FosterCnd}). Its proof is also  deferred to Appendix \ref{ApndB}.
\begin{lemma} \label{alpha} 
Using the notation of Lemma \ref{phiharm}, for any {$\alpha\geq \mu+\theta$,} {it holds that $\lim_{j\rightarrow\infty}{{\varphi}(j)}=\infty$}.
\end{lemma}

Finally, we can prove that the Markov chain {$Z$} is {positive} recurrent.
\begin{proof}[Proof of Theorem \ref{recurrence}] Consider the function $\varphi(\hat{y})=\varphi((y_1,j\pm1,y_2,j))=\varphi(j)$ as given by Equation (\ref{phisupharm}). Then, by the proof of Lemma \ref{phiharm}, we know that,
\begin{align} \label{rel:used:pos:recurr}
		{\varphi(j-1)(1-p_{\mathbf{N^*}})+\varphi(j+1)p_{\mathbf{1}}=\varphi(j)}.
\end{align}
Take any $\epsilon\in(0,1)$ {and define} $h(\hat{y})=\varphi(\hat{y})-\epsilon/(p_{\mathbf{1}}-p_{\mathbf{N^*}})$, and $h(j)=h(\hat{y})$, for any $\hat{y}=(y_1,j\pm1,y_2,j)$. Notice that $p_{\mathbf{1}}\more p_{\mathbf{N^*}}$ and, by the Lemma \ref{alpha}, $h(j)\rightarrow\infty$ as $j\rightarrow\infty$. Let $\Theta\in\R$ be such that for $j\more \Theta$, we have that $\varphi(j)\more\epsilon/(p_{\mathbf{1}}-p_{\mathbf{N^*}})$, and let also $F=\{\hat{y}\in\Xi: \hat{y}=(y_1,z_1,y_2,z_2),\ z_2\leq \Theta+1\}$. Notice that {$F$ is a finite set. From the definition of $h$ and the fact that $P(\hat{y},\hat{z})>0$ for only finitely many $\hat{z}$, it is clear that $h$ satisfies the first two Foster conditions shown in (\ref{FosterCnd}).} On the other hand, {following similar steps as in the proof of Lemma \ref{phiharm}, we have that,} for every $\hat{y}\notin {F}$,
\begin{align*}
\sum\limits_{\hat{z}\in\Xi}P(\hat{y},\hat{z})h(\hat{z})&=h(j-1)\Px(N\less\varsigma(y_2))+ h(j+1)\Px(\varsigma(y_2)\less N)\\
&=(\varphi(j-1)-\epsilon/(p_{\mathbf{1}}-p_{\mathbf{N^*}}))\Px(N\less\varsigma(y_2))+(\varphi(j+1)-\epsilon/(p_{\mathbf{1}}-p_{\mathbf{N^*}}))\Px(\varsigma(y_2)\less N)\\
&\leq(\varphi(j-1)-\epsilon/(p_{\mathbf{1}}-p_{\mathbf{N^*}}))(1-p_{\mathbf{N^*}})+(\varphi(j+1)-\epsilon/(p_{\mathbf{1}}-p_{\mathbf{N^*}}))p_{\mathbf{1}}\\
&=\varphi(j)-\epsilon(1-p_{\mathbf{N^*}}+p_{\mathbf{1}})/(p_{\mathbf{1}}-p_{\mathbf{N^*}})\\
&=h(j)-\epsilon=h(\hat{y})-\epsilon.
\end{align*}
This proves the last Foster condition given in (\ref{FosterCnd}) and the fact that $Z$ is positive recurrent follows.
\end{proof}

Once we have proved that $Z$ satisfies the hypothesis of Theorem \ref{LLNHarris}, we now introduce the functions on which the theorem is applied. For any $\hat{x}=(x_0,c_0,x_1,c_1)\in\Xi$, let 
\begin{align*}
	f(\hat{x})&:=\Ex(\tau_1|\hat{x}):=\Ex\left(\left.\tau_1\right|(\tilde{x}_0,\zeta_0)=(x_0,c_0),\ (\tilde{x}_1,\zeta_1)=(x_1,c_1)\right),\\
	g_t(\hat{x})&:=\Px(\tau_1\more t\left.\right|\hat{x}):=\Px(\tau_1\more t|(\tilde{x}_0,\zeta_0)=(x_0,c_0),(\tilde{x}_1,\zeta_1)=(x_1,c_1)).
\end{align*}
{We have the following result, whose proof is deferred to Appendix \ref{ApndB}}:
\begin{lemma}\label{FntnessIntg}
	Suppose that the conditions of {Theorem} \ref{recurrence} hold and let $\pi$ be the invariance probability of the chain $Z$. Then, $P_{*}$-a.s.,
	\begin{align} \label{CondLLN}
{\rm (i)}\;\;\lim\limits_{n\rightarrow\infty}\frac{1}{n}\sum_{k=1}^n f\left(Z_{k}\right)=\Ex_\pi(\tau_1),\quad
\quad %\hbox{a.s.}-P_*,
{\rm (ii)}\;\;\lim\limits_{n\rightarrow\infty}\frac{1}{n}\sum\limits_{k=1}^n g_{t}\left(Z_{k}\right)&=\Px_\pi(\tau_1\more t). %\label{CondLLNCorr2}
\end{align}
\end{lemma}
{In order to obtain the LLN for {the interarrival times $\{\tau_i\}_{i\geq{}1}$}, we shall} show that the Laplace transform of the random variables {$T_n=\tau_1+\ldots+\tau_n$, properly scaled,} converges to the Laplace transform of a random variable $T$, for which we need the following:

\begin{prop} \label{prop1}
For $u\in\R_{+}$ and $\hat{x}=(x_0,c_0,x_1,c_1)\in\Xi$, define the functions 
\begin{align*}
	G(u|\hat{x})&:=\Ex\left(\left.e^{-u\tau_{1}}
	\right|\left(\tilde{x}_{0},\zeta_{0},\tilde{x}_{1},\zeta_{1}\right)=(x_0,c_0,x_1,c_1)\right),\\
	\kappa(\hat{x})&:=\Ex\left(\left.-\tau_1\right|\left(\tilde{x}_{0},\zeta_{0},\tilde{x}_{1},\zeta_{1}\right)=(x_0,c_0,x_1,c_1)\right).
\end{align*}
Then, under the assumption {of Proposition} \ref{alpha},  for any $u\in[0,\infty]$,
\begin{equation}\label{WWPH}
	{\lim\limits_{n\rightarrow\infty}\sum\limits_{k=1}^{n}\ln G\left(\left.\frac{u}{n}\right|\tilde{x}_{k-1},\zeta_{k-1},\tilde{x}_{k},\zeta_{k}\right)=-u\Ex_\pi(\tau_1),\quad\quad{P_* -{a.s.}}},
\end{equation}
where $\pi$ is the stationary measure of the Markov chain $\{Z_n\}_{n\geq0}$.
\end{prop}
\begin{proof} First note that the statement is trivial for $u=0$. By (\ref{expectau}), $\tau_1\less\infty$ a.s., thus, $\Ex\left(\left.e^{-u\tau_{1}}
	\right|\tilde{x}_{0},\zeta_{0},\tilde{x}_{1},\zeta_{1}\right)\more0$ a.s. Assume now that $u\in(0,\infty)$, then, by Jensen's inequality,
\begin{align*}
	\frac{1}{u}\ln G(u|\tilde{x}_{0},\zeta_{0},\tilde{x}_{1},\zeta_{1})\geq \frac{1}{u}\Ex\left(\left.\ln e^{-u\tau_{1}}
	\right|\tilde{x}_{0},\zeta_{0},\tilde{x}_{1},\zeta_{1}\right)=\kappa(\tilde{x}_{0},\zeta_{0},\tilde{x}_{1},\zeta_{1}).
\end{align*}
Therefore, $\ln G\left.\left(\frac{u}{n}\right|\tilde{x}_{0},\zeta_{0},\tilde{x}_{1},\zeta_{1}\right)\geq \frac{u}{n}\kappa(\tilde{x}_{0},\zeta_{0},\tilde{x}_{1},\zeta_{1})$ and, thus,
\begin{equation}\label{lowbddg1}
\sum\limits_{k=1}^n \ln G\left.\left(\frac{u}{n}\right|\tilde{x}_{k-1},\zeta_{k-1},\tilde{x}_{k},\zeta_{k}\right)\geq \frac{u}{n}\sum\limits_{k=1}^n \kappa(\tilde{x}_{k-1},\zeta_{k-1},\tilde{x}_{k},\zeta_{k}),
\end{equation}
which, by {Eq. (\ref{CondLLN}-i),} implies that,
\[\liminf\limits_{n\rightarrow\infty}\sum\limits_{k=1}^n \ln G\left.\left(\frac{u}{n}\right|\tilde{x}_{k-1},\zeta_{k-1},\tilde{x}_{k},\zeta_{k}\right)\geq \liminf\limits_{n\rightarrow\infty}\frac{u}{n}\sum\limits_{k=1}^n \kappa(\tilde{x}_{k-1},\zeta_{k-1},\tilde{x}_{k},\zeta_{k})=-u\Ex_\pi(\tau_1)\quad\quad P_*-\hbox{a.s.}.\]
Next, note that
\begin{align}\nonumber
	{\frac{1}{u}\ln G(u|\tilde{x}_{0},\zeta_{0},\tilde{x}_{1},\zeta_{1})
	\leq \Ex\left( \left.\frac{e^{-u\tau_{1}} -1}{u}\right|\tilde{x}_{0},\zeta_{0},\tilde{x}_{1},\zeta_{1}\right)=\int_0^\infty -e^{-ut}\Px\left[\tau_1\more t\left.\right|\tilde{x}_{0},\zeta_{0},\tilde{x}_{1},\zeta_{1}\right]dt},
\end{align}
{where for the first inequality we used that $\ln(x)\leq x-1$, for $x>0$, and for the last equality we used the the identity $\Ex(g(X))=g(0)+\int_0^\infty g'(t)P[X\more t]dt${, which is}  valid for any positive random variable $X$ and monotonic differentiable function $g:[0,\infty)\rightarrow\R$.} Therefore, we have that:
	\[
	\ln G\left.\left(\frac{u}{n}\right|\tilde{x}_{0},\zeta_{0},\tilde{x}_{1},\zeta_{1}\right)\leq \frac{u}{n}\int_0^\infty -e^{-\frac{u}{n}t}\Px\left[\tau_1\more t\left.\right|\tilde{x}_{0},\zeta_{0},\tilde{x}_{1},\zeta_{1}\right]dt.
	\]
This last inequality, Fatou's Lemma, and {Eq.~(\ref{CondLLN}-i)} yield,
\small
\begin{align}\label{lowbddg2}\nonumber
\limsup\limits_{n\rightarrow\infty}\sum\limits_{k=1}^n \ln G\left.\left(\frac{u}{n}\right|\tilde{x}_{k-1},\zeta_{k-1},\tilde{x}_{k},\zeta_{k}\right)&\leq \limsup\limits_{n\rightarrow\infty}\sum\limits_{k=1}^n -\frac{u}{n}\int_0^\infty e^{-\frac{u}{n}t}\Px\left(\tau_1\more t\left.\right|\tilde{x}_{k-1},\zeta_{k-1},\tilde{x}_{k},\zeta_{k}\right)dt\\ \nonumber
			&\leq\int_0^\infty \limsup\limits_{n\rightarrow\infty}\left(-ue^{-\frac{u}{n}t}\right)\left(\frac{1}{n}\sum\limits_{k=1}^n \Px\left(\tau_1\more t\left.\right|\tilde{x}_{k-1},\zeta_{k-1},\tilde{x}_{k},\zeta_{k}\right)\right)dt\\
			&=\int_0^\infty-u \Px_\pi\left(\tau_1\more t\right)dt=-u\Ex_\pi(\tau_1),\quad\quad P_*-\hbox{a.s.}.
\end{align}
\normalsize
{Together (\ref{lowbddg1}) and (\ref{lowbddg2}) imply (\ref{WWPH}).}
\end{proof}

{We are now ready to show the main result of this section.}
\begin{thm}
Under the assumptions of Proposition \ref{prop1}, we have
\begin{equation} \label{LLNtau}
\frac{1}{n}\sum\limits_{k=1}^n \tau_k\stackrel{\Px}{\rightarrow} \Ex_\pi(\tau_1),\qquad\qquad\text{as }n\rightarrow\infty,
\end{equation}
where $\pi$ is the stationary measure of the Markov chain $\{Z_n\}_{n\geq0}$.
\end{thm}

\begin{proof}
Let $\varphi_{n}(u):=\Ex_\pi\left(e^{-u n^{-1}\sum_{k=1}^{n}\tau_{k}}\right)$ and $\calF_{n}:=\sigma((\tilde{x}_{k},\zeta_{k}):k\leq{}n)$. {By the conditional independence in (\ref{MIRNL}),  
\begin{align*}
\varphi_{n}(u)
	&= \Ex\left(\prod\limits_{k=1}^n \Ex\left(\left.e^{-u n^{-1} \tau_k}\right|\calF_n\right)\right)= \Ex\left(\prod\limits_{k=1}^n \Ex\left(\left.e^{-u n^{-1} \tau_k}\right|\tilde{x}_{k-1},\zeta_{k-1},\tilde{x}_{k},\zeta_{k}\right)\right)= {\Ex\left(e^{\sum_{k=1}^{n}\ln G\left.\left(\frac{u}{n}\right|\tilde{x}_{k-1},\zeta_{k-1},\tilde{x}_{k},\zeta_{k}\right)}\right)}.
\end{align*}
Since} for any positive $x$, $\ln(x)\leq x-1$, 
\[\ln G\left.\left(\frac{u}{n}\right|\tilde{x}_{0},\zeta_{0},\tilde{x}_{1},\zeta_{1}\right)\leq G\left.\left(\frac{u}{n}\right|\tilde{x}_{0},\zeta_{0},\tilde{x}_{1},\zeta_{1}\right) -1=\Ex\left(\left.e^{-\frac{u}{n}\tau_{1}}-1\right|\tilde{x}_{0},\zeta_{0},\tilde{x}_{1},\zeta_{1}\right)\leq 0,
\]
and, therefore, for every $n$,
\[
\exp\left\{\sum_{k=1}^{n}\ln G\left(\left.\frac{u}{n}\right|\tilde{x}_{k-1},\zeta_{k-1},\tilde{x}_{k},\zeta_{k}\right)\right\}\leq 1
\]
and, by Dominated Convergence Theorem and Proposition \ref{prop1}, we get
\[
\lim\limits_{n\rightarrow\infty}\varphi_{n}(u) = \lim\limits_{n\rightarrow\infty}\Ex\left(e^{\frac{-u}{n}\sum_{k=1}^{n}\tau_{k}}\right)= e^{-u\Ex_\pi(\tau_1)}.
\]
Finally, since $\tau_k$ is supported on the positive numbers, by the continuity theorem for Laplace transforms (see theorem 2 in section XIII.1 in \cite{FellerII}), we {obtain (\ref{LLNtau})}.
\end{proof}

\subsection{Long-run dynamics of the price process}\label{longrunsec}
In this section, {we obtain {a diffusive approximation for the} dynamics of the midprice process of the model defined in Section \ref{Sec:Mod2}}. Throughout, {$s_t$ denotes the {stock's} midprice} at time {$t\in[0,\infty)$}, while {$\tau_n$} represents the time elapsed between the $(n-1)-$th and the $n-$th price change as described in {Section} \ref{Sec:Mod2}. Let $\{{\tilde{u}}_n\}_{n\geq1}$ be the sequence of midprice changes. Clearly, our assumptions for the LOB dynamics described in Section \ref{Sec:Mod2} imply that
$
{\tilde{u}}_n\in\left\{-1/2,1/2\right\}.
$
{It is also easy to see that the midprice process is given by
\begin{equation}\label{priceprocess}
s_t:=s_0+\sum_{j=1}^{N_t} {\tilde{u}}_j,
\end{equation}
where hereafter $N_t:=\max\{n\left.\right| \tau_1+\ldots+\tau_n\leq t\}$ denotes the number of price changes up to time $t$. {In this section, we establish the relation (\ref{GlAsMp}), for some constants $\sigma>0$ and $m$.}

Recall from Section \ref{Sec:Mod2} that {$U_n:=(\tilde{x}_{n},\zeta_{n})=((x_{n}^{b},x_{n}^{a}),\zeta_{n})$}, the number of orders in the level I of the book and the spread after the $n-$th price change, is a Markov chain  {(cf. Eqs.~(\ref{MPI})-(\ref{MPII}))}. Also, recall that, for $i\geq0$, $Q^{a,i}(x)$ and $Q^{b,i}(x)$ are independent continuous-time Markov processes with common generator defined by (\ref{TrnsMtrQ}). Define $\varsigma_n:=\inf\{t\more0:{Q_t^{a,n}({\tilde{x}_{n-1}^{a}})\wedge Q_t^{b,n}({\tilde{x}_{n-1}^{b}})}=0\}$ {and also consider} the following {events:} 
\begin{align*}
A_n&=\left\{Q^{\,{b,n}}_{{\varsigma_{n}}}({\tilde{x}^{\, b}_{n-1}})=0,\; \zeta_{n-1}=1\quad\text{ or }\quad Q^{\, {b,n}}_{{\varsigma_{n}}}({\tilde{x}^{\, b}_{n-1}})=0,\;\zeta_{n-1}\more1,\;
{L_{n}(\zeta_{n-1})}\geq{\varsigma_{n}}\right\},\\ 
B_n&=\left\{{L_{n}(\zeta_{n-1})}\less{\varsigma_{n}},\; {L_{n}(\zeta_{n-1})=L^{a}_{n}(\zeta_{n-1})},\;\zeta_{n-1}\more1\right\},\\ 
C_n&= \left\{Q^{\,{a,n}}_{{\varsigma_{n}}}(\tilde{x}^{\,a}_{n-1})=0,\;
\zeta_{n-1}=1\quad\text{ or }\quad {L_{n}(\zeta_{n-1})}\geq{\varsigma_{n}},\; Q^{\, {a,n}}_{{\varsigma_{n}}}({\tilde{x}^{\,a}_{n-1}})=0,\;\zeta_{n-1}\more1\right\},\\ 
D_n&= \left\{{L_{n}(\zeta_{n-1})\less\varsigma_{n},\; {L_{n}(\zeta_{n-1})=L^{b}_{n}(\zeta_{n-1})},\; \zeta_{n-1}\more1}\right\},
\end{align*}
where $\{L^{a}_{k}(\zeta)\}_{k,\zeta\in\mathbb{Z}_{+}},\{L^{b}_{k}(\zeta)\}_{k,\zeta\in\mathbb{Z}_{+}}$ and $\{L_{k}(\zeta)\}_{k,\zeta\in\mathbb{Z}_{+}}$ are the random variables defined in Section \ref{Sec:Mod2}.

A positive price change would occur at time {$T_n$} if, either the ask queue got depleted (event $A_n$ above) or a new queue arrived at the bid side (event $B_n$). Similarly, a negative price change would occur if either the bid queue got depleted (event $C_n$) or a new queue arrived at the ask side (event $D_n$). {Therefore,}
\begin{equation}\label{defnun}
{{\tilde{u}}}_n:=\frac{1}{2}\left[\indicator{A_n}+\indicator{B_n}\right]-\frac{1}{2}\left[\indicator{C_n}+\indicator{D_n}\right],
\end{equation}
{represents} the $n-$th price change, for $n\geq{}1$.}

As in the preceding section, {an important step for analyzing} the price changes would be {to express those in terms of an appropriate Markov chain}. Let $\Lambda:=\{\overline{z}=(y_1,c_1,u):y_1\in\Omega_{N^*}^2,c_1\in{\mathbb{Z}_{+}},u\in\{-1/2,1/2\}\}$ and
\begin{equation}\label{Vndefn}
V_n:=(\tilde{x}_{n},\zeta_{n}, {{\tilde{u}}}_n),
\end{equation}
for $n\geq1$. {Note that $V:=\{V_n\}_{n\geq0}$ is a Markov chain over $\Lambda$} since $\tilde{x}_n, \zeta_n$ and ${{\tilde{u}}}_n$ depend only on $(\tilde{x}_{n-1},\zeta_{n-1}$). Moreover, one can see that the states of $V$ communicate to one another and, thus, $V$ is irreducible. Also, provided that the assumptions of Lemmas \ref{phiharm} and \ref{alpha} hold, {one} can prove that $V$ is recurrent, similarly to the proof of Theorem \ref{recurrence}, and {$V$ will then be Harris recurrent due to the} countability of $V$'s state space. As a consequence, $V$ would also be positive Harris. Hereafter, we denote the stationary measure {and the {transition probabilities} of $V$ by $\nu$ and {$P^{ext}(\bar{y},\bar{z})$}, respectively.}

{As mentioned above}, our main goal is to establish the {coarse-grained behavior} of the price process (\ref{priceprocess}). In order to do so, we first analyze the convergence of the process $W^n:=\sum_{j=1}^{n} {\tilde{u}}_j$, properly rescaled. To this end, {the following Functional Central Limit Theorem (FCLT) for Markov Chains on a countable state space {will be useful}:}

\begin{thm}[{\cite{Meyn}, Theorem 17.4.4}] \label{FCLTHarris}
Suppose that $\{V_n\}_{n\geq0}$ is positive Harris {on a  countable state space $\Lambda$ with transition and stationary probability measures $P^{ext}$ and $\nu$, respectively. Let $h$ be} a function on $\Lambda$ for which a solution $\hat{h}$ to the Poisson equation, 
\begin{equation}\label{poissoneqn}
\hat{h}-P^{ext}\hat{h}=h-\nu(h),
\end{equation}
exists with $\nu(\hat{h}^2)\less\infty$. {Consider the partial sums of the centered functional $\bar{h}(V_k):=h(V_k)-\nu(h)$,}
\begin{equation}\label{Sndefn}
S_n(\bar{h}):=\sum\limits_{k=1}^n {\bar{h}(V_k)},
\end{equation}
and let $r_n(t)$ be the continuous piece-wise linear function that interpolates the values of {$\left\{S_{n}(\bar{h})\right\}_{n\geq{}0}$}; i.e., 
\begin{equation}\label{rndefn}
r_n(t):=S_{\lfloor nt\rfloor}(\bar{h}) + (nt-\lfloor nt\rfloor)\left[S_{\lfloor nt\rfloor+1}(\bar{h})-S_{\lfloor nt\rfloor}(\bar{h})\right].
\end{equation}
Then, if the constant 
\begin{equation}\label{sigmaFCLT}
\gamma^2(h):={\nu\left(\hat{h}^2-({P^{ext}}\hat{h})^2\right)}
\end{equation}
is positive, it holds that,
\begin{equation}\label{InvTh2}
	\left\{\frac{r_n(t)}{\sqrt{n{\gamma^2(h)}}}\right\}_{t\geq0}\Rightarrow \left\{W_t\right\}_{t\geq0},\quad {n\rightarrow\infty.}
\end{equation}
\end{thm}

\begin{rem} By taking $t=1$, it follows that
\[
\sqrt{n}\left(\frac{1}{n}\sum\limits_{k=1}^n h(V_k)-\nu(h)\right)\Rightarrow \mathcal{N}(0,\gamma^2(h)).
\]
\end{rem}
\noindent
{The proof of the next result} is deferred to {Appendix} \ref{ApndB}.

\begin{thm} \label{FCLTh} {Let {$V:=\{V_n\}_{n\geq{}1}=\{(\tilde{x}_{n},\zeta_{n}, {{\tilde{u}}}_n)\}_{n\geq{}1}$} be the Markov chain {defined} on (\ref{Vndefn}) with stationary {probability} measure $\nu$. Then, for $h:\Lambda\rightarrow\R$ given by $h(x,c,u)= u$,  }there exists a solution to the Poisson equation $\hat{h}$ with $\nu(\hat{h}^2)\less\infty$. Furthermore, the invariance {principle} (\ref{InvTh2}) holds true {and the variance $\gamma^{2}(h)$ admits the representation
\begin{equation}\label{exprsigma} 
\gamma^2(h)=\Ex_{\nu}\left(\bar{h}^2({V_1})\right)+2\sum\limits_{k={2}}^\infty \Ex_\nu\left(\bar{h}({V_1})\bar{h}(V_k)\right),
\end{equation}
where $\bar{h}=h-\nu(h)$ and the sum converges absolutely.}
\end{thm}
In the following, we will write {$f_{n}\stackrel{\Px}{\sim} g_{n}$ {if $\lim_{n\rightarrow\infty}f_{n}/g_{n}=1$, in probability.}
The following {result} is the final ingredient towards (\ref{GlAsMp}):
\begin{lemma}\label{Ntnassympthm2}
 Using the notation of {Section} \ref{LLNmodtau},
\begin{equation}\label{NdRsl}
N_{{tn}}\stackrel{\Px}{\sim}\frac{tn}{\Ex_{\pi}(\tau_1)},\qquad  \text{as }\,n\rightarrow\infty,
\end{equation}
where {we recall that}
$N_t=\max\{n\left.\right| \tau_1+\ldots+\tau_n\leq t\}$ and $\pi$ is the stationary measure of the chain $Z_n=(\tilde{x}_{n-1},\zeta_{n-1},\tilde{x}_{n},\zeta_{n})$, whose existence is {guaranteed by Theorem \ref{recurrence}}.
\end{lemma} 
\begin{proof}
{Throughout, let $t_n:=tn$.} Since $N_{t_n}$ denote the number of price changes up to time $t_n$,
$$\frac{\tau_1+\ldots+\tau_{N_{t_n}}}{N_{t_n}} \leq \frac{t_n}{N_{t_n}} \less\frac{\tau_1+\ldots+\tau_{N_{t_n}+1}}{N_{t_n}}$$
and, thus, by (\ref{LLNtau}), as $n\rightarrow\infty$, $\frac{t_n}{N_{t_n}}\stackrel{\Px}{\rightarrow}\Ex_{\pi}(\tau_1)$, {which in turn implies (\ref{NdRsl}).}
\end{proof}

Finally, we can state the main result on this section.

\begin{thm} \label{Thm:finalFCLT}
Let {$\{s_t\}_{t\geq{}0}$} be the price process as defined {in Eq.~(\ref{priceprocess})}. Then, 
\begin{equation}\label{Nth2}
\left\{\sqrt{n}\left(\frac{s_{{tn}}}{n}- \frac{\nu(h)}{\Ex_\pi(\tau_1)}t\right)\right\}_{t\geq0}{\Rightarrow}\,\left\{\gamma(h)W_t\right\}_{t\geq0},\quad\text{as }n\rightarrow\infty,
\end{equation}
where {the variance $\gamma^{2}(h)$ is given as in Eq.~(\ref{exprsigma}).}
\end{thm}
\begin{proof} {Throughout, let $t_n:=tn$.} {Let us recall that $s_{t}=s_0+\sum\limits_{j=1}^{N_t}{\tilde{u}}_j$ and ${\tilde{u}}_n=h(V_n)$,} for the Markov chain $\{V_n\}_{n\geq0}$ and $h:\Lambda\rightarrow\R$ {given} by {$h(y,c,u)=u$}.
Now, we decompose the process {$\bar{s}_{t_{n}}:=n^{1/2}\left(s_{t_n}/n -t\nu(h)/\Ex_\pi(\tau_1)\right)$} as:
\small
\[
	{\bar{s}_{t_n}}= \underbrace{\frac{s_0}{\sqrt{n}}}_{\hbox{I}_n}+\underbrace{\frac{1}{\sqrt{n}}\sum\limits_{j=1}^{[tn/\Ex_{\pi}(\tau_1)]}\left({\tilde{u}}_j- \nu(h)\right)}_{\hbox{II}_n} + \underbrace{\left(\frac{1}{\sqrt{n}}\sum\limits_{j=1}^{N_{t_n}}{\tilde{u}}_j-\frac{1}{\sqrt{n}}\sum\limits_{j=1}^{[tn/\Ex_{\pi}(\tau_1)]}{\tilde{u}}_j\right)}_{\hbox{III}_n}+\underbrace{\left(\frac{1}{\sqrt{n}}\sum\limits_{j=1}^{[tn/\Ex_{\pi}(\tau_1)]}\nu(h)-\sqrt{n}\frac{t\nu(h)}{\Ex_\pi(\tau_1)}\right)}_{\hbox{IV}_n} ,
\]
\normalsize
\noindent
where, as in Theorem \ref{Ntnassympthm2}, {$\nu$} is the stationary measure of the Markov chain $\{V_n\}_{n\geq0}$. As $n\rightarrow\infty$, clearly, I$_n\Rightarrow0$.
Also, by Theorem \ref{FCLTh},
\begin{align*}
\hbox{II}_n&\Rightarrow \gamma(h)^2 W_t,
\end{align*}
where $\gamma^2(h)$ is given by  {Eq.~(\ref{exprsigma})}. Now, since ${\tilde{u}}_j\in\left\{\frac{1}{2},-\frac{1}{2}\right\}$, for any $\epsilon\more0$,
\begin{align*}
\Px\left(\left|\sum\limits_{j=1}^{N_{t_n}}{\tilde{u}}_j-\sum\limits_{j=1}^{[tn/\Ex_{\pi}(\tau_1)]}{\tilde{u}}_j\right|\geq\epsilon\sqrt{n}\right)&\leq\Px\left(\left|\sum\limits_{j=N_{t_n}\wedge[tn/\Ex_{\pi}(\tau_1)]}^{N_{t_n}\vee[tn/\Ex_{\pi}(\tau_1)]}{\tilde{u}}_j\right|\geq\epsilon\sqrt{n}\right)\\
	&\leq \Px\left(\frac{1}{2}\left|N_{t_n}-[tn/\Ex_{\pi}(\tau_1)]\right|\geq\epsilon\sqrt{n}\right)\\
	&\leq \Px\left(\left|\frac{N_{t_n}}{[tn/\Ex_{\pi}(\tau_1)]}-1\right|\geq\frac{2\epsilon\sqrt{n}}{[tn/\Ex_{\pi}(\tau_1)]}\right),
\end{align*}
which, by {Proposition \ref{Ntnassympthm2}}, converges to 0 as $n\rightarrow\infty$. Thus, III$_n$ converges to 0 in probability. Finally, since $\hbox{IV}_n=\nu(h)\frac{[tn/\Ex_{\pi}(\tau_1)]}{\sqrt{n}}-\sqrt{n}\frac{t\nu(h)}{\Ex_\pi(\tau_1)}$ is such that $0\leq-\hbox{IV}_n\less\frac{\nu(h)}{\sqrt{n}}$, it follows that IV$_n\rightarrow0$, as $n\rightarrow\infty$, {and, thus, we conclude  (\ref{Nth2}).} 
\end{proof}

\section{Computation of Some {LOB} Features of Interest}\label{CmpImpQnt}

In this section we develop some numerical tools to evaluate some {LOB model features} of practical relevance such as the distribution of the time span between price changes, the probability of a price increase, and the probability of two consecutive price increments. The proposed method is based on an explicit characterization of the joint distribution of the time and position at which a certain two-dimensional Markov chain starting in the first quadrant hits the coordinate axes. The developed tools will also be used in Section \ref{Implementation} to {devise} an efficient simulation algorithm for the midprice dynamics of the order book.

Recall that  {$\bar\Omega_{{N^{*}}}:=\{0,1,2,\ldots,{N^{*}}\}$ and $\Omega_{N^{*}}:=\{1,2,\ldots,{N^{*}}\}$. Throughout this section, we let $\left\{Y(x,y)\right\}_{(x,y)\in\Omega^{2}_{N^{*}}}$ be a collection of independent processes such that, for each $(x,y)\in\Omega^{2}_{N^{*}}$,
\begin{equation}\label{ChainY}
	{Y(x,y):=\left\{Y_t(x,y)\right\}_{t\in\mathbb{N}}:=\left\{\left(Q_t^{a,0}\left(x\right),Q_t^{b,0}\left(y\right)\right)\right\}_{t\in\mathbb{N}}},
\end{equation}
where $Q^{a,0}(x):=\{Q^{a,0}_{t}(x)\}_{t\geq{}0}$ and $Q^{b,0}(x):=\{{Q}^{b,0}_{t}(x)\}_{t\geq{}0}$
are defined as in Section \ref{Sec:Mod2} (see Eq.~(\ref{DfnQaQb})).}
{We also set
\begin{align*}
	&\mathscr{A}_A:=\{(0,1),(0,2),\ldots,{(0,N^*)}\},\quad 
	\mathscr{A}_B:=\{(1,0),(2,0),\ldots,(N^*,0)\},\quad \mathscr{A}:= \mathscr{A}_A\cup  \mathscr{A}_B,\\
 	&\varsigma(x,y):=\inf\{t\more0: {Y_{t}(x,y)\in \mathscr{A}}\}, \quad L:=L^{a}\wedge L^{b},
\end{align*}
where $L^{a}$ and $L^{b}$ are independent exponential variables with parameter $\alpha$. These variables are meant to represent the times for a new set of orders to arrive at the ask and bid side, respectively. Finally, $\stackrel{\mathcal{D}}{=}$ denotes equality in distribution.}

\subsection{{Distribution of the duration between price changes}}\label{sec:Distr:time}
Here, we develop a numerical method to find the distribution of the {first price change time $\tau_{1}$ given that, initially at time $0$, there are $x$ orders at the bid, $y$ at the ask, and the spread is $z$.
To this end, we first  compute the joint distribution of the vector $(\varsigma(x,y),Y_{\varsigma(x,y)}(x,y))$. 
This} is obtained via the following two {lemmas}, whose proofs can be found in the Appendix {\ref{ApndB}}
\begin{lemma} \label{jointdist} 
{Suppose that, for each fixed $\ax:=(\bar{a}_{1},\bar{a}_{2})\in\mathscr{A}$, $u_{\bar{a}}:[0,T]\times\bar\Omega_{{N^{*}}}\to\mathbb{R}$ satisfies the {following system of differential equations}:
\begin{equation}\label{PDEu}
%\left\{
\begin{array}{ll}
\left.\left(-\frac{\partial}{\partial t} + \mathscr{L}\right) u_{{\ax}}(t,x,y)\right|_{t=T-r}=0, & \text{for}\quad {0\leq r < T}, \quad (x,y)\in\Omega_{N^{*}}\\
u_{{\ax}}(T-r,x,y)=\indicator{(x,y)=\ax}, & \text{for}\quad  0\leq r\leq T,\quad (x,y)\in\mathscr{A},\\
u_{{\ax}}(0,x,y)=\indicator{(x,y)=\ax}, & \text{for}\quad  (x,y)\in{\bar\Omega_{N^{*}}^{2}},\end{array}
%\right.
\end{equation}
where $\mathscr{L}u(t,x,y)$ is {the} finite difference operator given by
\begin{equation} \label{uGenerator}
\mathscr{L}u(t,x,y)=\left\{\begin{array}{ll}
\lambda (u_1^++u_2^+) +\up(u_1^-+u_2^-) -2(\lambda+\up)u , & (x,y)\in\{1,2,\ldots,{N^{*}}-1\}^2,\\
\lambda u_2^+ +\up(u_1^-+u_2^-) -(\lambda+2\up)u,& x={N^{*}},\ y\in\{1,2,\ldots,{N^{*}}-1\},\\
\lambda u_1^+ +\up(u_1^-+u_2^-) -(\lambda+2\up)u,& x\in\{1,2,\ldots,{N^{*}}-1\},\ y={N^{*}},\\
\up(u_1^-+u_2^-) -2\up u, & (x,y)=({N^{*}},{N^{*}}),\\
0,& (x,y)\in\mathscr{A},
\end{array}\right.
\end{equation}
and $u_1^+=u(t,x+1,y)$, $u_2^+=u(t,x,y+1)$, $u_1^-=u(t,x-1,y)$, $u_2^-=u(t,x,y-1)$, and $u=u(t,x,y)$.
Then, for ${t}\more0$, {$(x,y)\in\bar\Omega_{N^{*}}^{2}$}, and $\ax:=(\bar{a}_{1},\bar{a}_{2})\in\mathscr{A}$,}
\begin{equation}\label{RPDJH}
	{u_{\bar{a}}(t,x,y):=\Px\left[\varsigma(x,y)\leq {t}, {Y_{\varsigma(x,y)}(x,y)}=\ax\right]}.
\end{equation}
\end{lemma}

The next {result proves the existence of {a} solution $u$ to the system (\ref{PDEu}) by giving an explicit {representation} of $u$ in terms of the eigenvalues and eigenvectors of a certain finite difference operator. As a result, we obtain as well an explicit formulation of the joint distribution of $(\varsigma(x,y),Y_{\varsigma(x,y)}(x,y))$}.  Below, we let
\[
\overline{a+1}:=\left\{\begin{array}{rcl} {\bar{a}+(0,1),} & if & \ax\in\{(1,0),(2,0),\ldots,({N^{*}},0)\}\\ {\bar{a}+(1,0),}& if & \ax\in\{(0,1),(0,2),\ldots,(0,{N^{*}})\}.\end{array}\right.
\]
\begin{prop} \label{lemmausol}
Let $\Delta$ be the symmetric finite difference operator defined for functions {$w:\bar{\Omega}^{2}_{N^{*}}\to\mathbb{R}$} as
\begin{equation} \label{vGenerator}
\Delta {w(x,y)}=\left\{\begin{array}{ll}
w_1^++w_2^+ + w_1^-+w_2^- -4w , & (x,y)\in\{1,2,\ldots,{N^{*}}-1\}^2\\
w_2^+ +w_1^-+w_2^- -\left(4-\sqrt{\frac{\lambda}{\up}}\right)w,& x={N^{*}},\ y\in\{1,2,\ldots,{N^{*}}-1\}\\
w_1^+ +w_1^-+w_2^- -\left(4-\sqrt{\frac{\lambda}{\up}}\right)w,& x\in\{1,2,\ldots,{N^{*}}-1\},\ y={N^{*}}\\
w_1^-+w_2^- -\left(4-2\sqrt{\frac{\lambda}{\up}}\right)w, & (x,y)=({N^{*}},{N^{*}})\\
0,& (x,y)\in\mathscr{A}\end{array}\right.,
\end{equation}
where $w_1^+=w(x+1,y)$, $w_2^+=w(x,y+1)$, $w_1^-=w(x-1,y)$, $w_2^-={w(x,y-1)}$, and $w=w(x,y)$.
Let $\{\xi_k\}_{k=1}^{{N^{*}}^{2}}$ be the eigenvalues of $\Delta$ and $\{f_k(x,y)\}_{k=1}^{{N^{*}}^{2}}$ be their corresponding eigenvectors so that they constitute an orthonormal basis of $\mathbb{R}^{{N^{*}}^{2}}$. 
{For $\ax:=(\bar{a}_{1},\bar{a}_{2})\in\mathscr{A}$, let} $u_{\bar{a}}:[0,T]\times{\bar\Omega^{2}_{{N^{*}}}}\to\mathbb{R}$ be defined by 
\begin{equation}\label{uexplicit}
u_{{\bar{a}}}(t,x,y)=\left(\frac{\lambda}{\up}\right)^{\frac{\ax_1+\ax_2-x-y}{2}}\left[\sum\limits_{k=1}^{\ {N^{*}}^2}\frac{\sqrt{\lambda\up}f_k\left(\overline{a+1}\right)}{2(\lambda+\up)-\sqrt{\lambda\up}(4+\xi_k)}\left(1 - e^{-t\left[2(\lambda+\up)-(4+\xi_k)\sqrt{\lambda\up}\right]}\right)f_k(x,y){\indicator{(x,y)\in{\Omega}_{N^{*}}^{2}}}+\indicator{(x,y)=\ax}\right].
\end{equation}
Then, the function $u_{\bar{a}}$ satisfies the system of differential {equations} (\ref{PDEu}) {and, therefore, the identity (\ref{RPDJH}) holds true.} 
\end{prop}

\begin{rem}\label{Rem:Quot} 
We can rewrite Eq.~(\ref{uexplicit}) as:
\begin{equation*} %\label{uexplicit2}
u_{\bar{a}}(t,x,y):=\chi^{\frac{\ax_1+\ax_2-x-y}{2}}\sum\limits_{k=1}^{\ {N^{*}}^2}\frac{f_k\left(\overline{a+1}\right)}{2\left(\chi^{1/2}-\chi^{-1/2}\right)^2-\xi_k}\left(1 - e^{-2\lambda t \left[\left(\chi^{-1/2}-1\right)^2- \frac{\xi_k}{2}\chi^{-1/2}\right]}\right)f_k(x,y){\indicator{(x,y)\in{\Omega}_{N^{*}}^{2}}}+{\chi^{\frac{\ax_1+\ax_2-x-y}{2}}}\indicator{(x,y)=\ax},
\end{equation*}
where $\chi:=\lambda/\up$. 
The previous expression shows that, as {$t$} gets larger, the joint probability distribution {$P\left[\varsigma(x,y)\leq {t}, {Y_{\varsigma(x,y)}(x,y)}=\ax\right]$} depends on the parameters $\up$ and $\lambda$ mostly through the quotient $\chi=\lambda/\up$. Let us also point out that the eigenvalues of $\Delta$ can be proven to be non-positive and, thus, $u_{\bar{a}}(T,x,y)\in[0,1]$.
\end{rem}

We are now ready to compute the distribution $F_{\tau_{1}}\left(t|x,y,z\right):=\Px[\left.\tau_{1}\leq t\right|x_{0}^{a}=x,x_{0}^{b}=y,\zeta_{0}=z]$ of the time $\tau_{1}$ it takes for the price to change conditioned on the initial state of the book. For simplicity of notation, throughout $\tau(x,y,z)$ represents a random time such that $\Px(\tau(x,y,z)\leq{}t)=\Px[\left.\tau_{1}\leq t\right|x_{0}^{a}=x,x_{0}^{b}=y,\zeta_{0}=z]$, for any $t\geq{}0$. It is clear that {$\tau(x,y,1)\stackrel{\mathcal{D}}{=}\varsigma(x,y)$} and, thus, from Eq.~(\ref{uexplicit}), for $(x,y)\notin\mathscr{A}$,
\begin{equation}\label{Distribution:tau:z1}
	{F_{\tau_{1}}\left(t|x,y,z\right)}=\left(\frac{\lambda}{\up}\right)^{-\frac{x+y}{2}}\sum\limits_{k=1}^{\ {N^{*}}^2}\frac{\sqrt{\lambda\up}\varsigma_k}{2(\lambda+\up)-\sqrt{\lambda\up}(4+\xi_k)}\left(1 - e^{-t(2(\lambda+\up)-(4+\xi_k)\sqrt{\lambda\up})}\right)f_k(x,y),
\end{equation}
where 
\begin{equation}\label{DfnSigmk}
	\varsigma_{k}:=\sum\limits_{\ax\in\mathscr{A}}\left(\frac{\lambda}{\up}\right)^{\frac{\ax_1+\ax_2}{2}}f_k\left(\overline{a+1}\right).
\end{equation}
On the other hand, for $z\geq2$, we have that {$\tau(x,y,z)\stackrel{\mathcal{D}}{=}\varsigma(x,y)\wedge L$, where as before $L$ represents the arrival time of a limit order within the spread.} Therefore, from the independence of $\varsigma(x,y)$ and ${L}$, {for any $z\geq{}2$ and $(x,y)\notin\mathscr{A}$,} 
\begin{align}\label{Distribution:tau:z2}
	{F_{\tau_{1}}(t|x,y,z)
		=\Px[L\leq{}t]+\Px[\varsigma(x,y)\leq{}t]\Px[L>t]=(1-e^{-2\alpha t})+F_{\tau_{1}}(t|x,y,1)e^{-2\alpha t}}.
\end{align}
The expressions (\ref{Distribution:tau:z1})-(\ref{Distribution:tau:z2}) provide {an} efficient numerical {method} to compute the distribution of the time span between price changes given some initial level I LOB setup. {The method is relatively efficient since} the main task in their evaluation is the computation of the eigenvalues $\{\xi_k\}_{k=1}^{{N^{*}}^{2}}$ and eigenvectors $\{f_k(x,y)\}_{k=1}^{{N^{*}}^{2}}$, which has to be done only once, for any $t\geq{}0$ and $z\in\{1,2,\dots\}$.}

\subsection{Probability of a price increase}\label{sec:PriceIncrease2}
We now consider the probability of a price increase conditioned on the current state of the order book:
\[
p(x,y,z):=\left.\Px\left[\text{Price increase }\right|x\text{ orders at Bid, }y\text{ orders at Ask, and a spread }z\right],\quad\text{for }(x,y)\notin\mathscr{A}.
\]
A price increase occurs if the {best ask} queue gets depleted or if a new set of orders arrives at the bid side. 
Recall from Lemma \ref{jointdist} that $u_{\bar{a}}(t,x,y):=P\left[\varsigma(x,y)\leq {t}, Y_{\varsigma(x,y)}(x,y)=\ax\right]$ has an explicit form given by {Eq.~(\ref{uexplicit})}. Set
\begin{equation}\label{uBdef}
	u_{B}(t,x,y):=\Px\left[\varsigma(x,y)\leq t , Y_{\varsigma(x,y)}\in \mathscr{A}_B\right]=\sum\limits_{\ax\in\mathscr{A}_B}u_{\ax}(t,x,y),
\end{equation}
and note that, if the spread is $z=1$,
\begin{equation}\label{pup:z1}
p(x,y,1)=u_{B}(\infty,x,y)=\left(\frac{\lambda}{\up}\right)^{-\frac{x+y}{2}}\sum\limits_{k=1}^{\ {N^{*}}^2}\frac{\sqrt{\lambda\up}\,\varsigma_{k,B}}{2(\lambda+\up)-\sqrt{\lambda\up}(4+\xi_k)}f_k(x,y),
\end{equation}
where
\[
	\varsigma_{k,B}:=\sum\limits_{\ax\in\mathscr{A}_{B}}\left(\frac{\lambda}{\up}\right)^{\frac{\ax_1+\ax_2}{2}}f_k\left(\overline{a+1}\right).
\]
In order to find $p(x,y,z)$ for $z\geq2$, {note that}
\begin{equation}\label{pup:z2}
p(x,y,z)=\Px\left[\varsigma(x,y)\leq {L},Y_{\varsigma(x,y)}(x,y)\in \mathscr{A}_B\right]+\Px\left[\varsigma(x,y)> {L},{L=L^{b}}\right]=:p_{1}(x,y)+p_{2}(x,y).
\end{equation}
By conditioning on ${L}$ and recalling that ${L}\sim {\rm exp}(2\alpha)$, 
\begin{align}\label{pup:z2I}\nonumber
p_{1}(x,y)
&=2\alpha \int_0^\infty u_{B}(t,x,y)e^{-2\alpha t} dt\\ 
&=\left(\frac{\lambda}{\up}\right)^{-\frac{x+y}{2}}\sum\limits_{k=1}^{\ {N^{*}}^2}\frac{\sqrt{\lambda\up}\,\varsigma_{k,B}}{2(\lambda+\up)-\sqrt{\lambda\up}(4+\xi_k)}\left(1 - \frac{2\alpha}{2(\lambda+\up+\alpha)-(4+\xi_k)\sqrt{\lambda\up}}\right)f_k(x,y).
\end{align}
For the second term, using the symmetry between ${L^{a}}$ and ${L^{b}}$, {$p_{2}(x,y)=\frac{1}{2}\Px\left[\varsigma(x,y)\geq N\right]$} and, thus, 
\begin{align} \nonumber
p_{2}(x,y)
		&=\frac{1}{2}\left(1-2\alpha \int_0^\infty \Px\left[\varsigma(x,y)\leq t\right]e^{-2\alpha t} dt\right)\\ 
		&={\frac{1}{2}\left(1- \left(\frac{\lambda}{\up}\right)^{-\frac{x+y}{2}}\sum\limits_{k=1}^{\ {N^{*}}^2}\frac{\sqrt{\lambda\up}\varsigma_k}{2(\lambda+\up)-\sqrt{\lambda\up}(4+\xi_k)}\left(1 - \frac{2\alpha}{2(\lambda+\up+\alpha)	-(4+\xi_k)\sqrt{\lambda\up}}\right)f_k(x,y)  \right),} \label{pup:II}
\end{align}
where $\varsigma_{k}$ is defined as in (\ref{DfnSigmk}). Again, once the  eigenvalues and eigenvectors of $\Delta$ have been computed, one can readily compute $p(x,y,z)$ via (\ref{pup:z2I})-(\ref{pup:II}), {for any $(x,y)\in \Omega_{N^{*}}^{2}$ and $z\in\mathbb{Z}_{+}$}.

\subsection{Probability of two consecutive price increments}
{Let $\hat{p}(x,y,z)$ be the probability of two consecutive increments in the price given that initially there were $x$ orders at the {best bid}, $y$ orders at the {best ask}, and a spread of $z$. These probabilities are highly dependent on the initial spread. {The case of an initial spread of 1 is relatively easier to analyze than any other spread due to the possibility of a new set of orders within the spread before the depletion of any of the level I queues. As will be shown below, in the latter situation, we will have to consider a probability of the form $\Px\left[{L}<\varsigma(x,y), Y_{{L}}(x,y)\in \{(1,j),\dots,(N^{*},j)\}\right]$, for any $j$. 
The aforementioned probability will be reformulated in terms of the solution to a certain initial value problem along the lines of Proposition \ref{lemmausol}.}

Recall that every time there is a price change, a new number of orders in the LOB side that got depleted is generated from a discrete distribution, {$f^{a}$ or $f^{b}$,} supported on $\{1,2,\ldots,N^*\}$, {depending on whether the best ask or bid queues got depleted. For simplicity, in what follows we assume that $f:=f^{a}=f^{b}$.} Denote {$H$ a random variable with distribution $f$}. In addition to the collection of random walks $\{Y(x,y)\}_{(x,y)\in{\Omega_{N^{*}}^{2}}}$ described at the beginning of Section \ref{sec:Distr:time}, we also need to consider another independent copy $\{\tilde{Y}(x,y)\}_{(x,y)\in{\Omega_{N^{*}}^{2}}}$ and fix $\tilde\varsigma(x,y):=\inf\{t\more0: \tilde{Y}_{t}(x,y)\in\mathscr{A}\}$. Similarly, in addition to {$(L^{a},L^{b})$}, we consider an independent copy {$(\tilde{L}^{a},\tilde{L}^{b})$} and fix {$\tilde{L}:=\tilde{L}^{a}\wedge \tilde{L}^{b}$}. 
We are ready to compute $\hat{p}(x,y,z)$. 

For $z=1$, {clearly,}
\begin{align*}
\hat{p}(x,y,1)&=\sum\limits_{i=1}^{N^*}\sum_{j=1}^{N^{*}}\Px\left[Y_{\varsigma(x,y)}(x,y)=(j,0), H=i, \tilde\varsigma(j,i)\leq {\tilde{L}}, \tilde{Y}_{\tilde\varsigma(j,i)}^2(j,i)\in \mathscr{A}_B\right]\\ 
			&\quad+ \sum\limits_{i=1}^{N^*}\sum_{j=1}^{N^{*}}\Px\left[Y_{\varsigma(x,y)}(x,y)=(j,0), H=i, \tilde\varsigma(j,i)\geq {\tilde{L}}, {\tilde{L}=\tilde{L}^{b}}\right]\\
			&=\sum\limits_{i=1}^{N^*}\sum_{j=1}^{N^{*}} u_{(j,0)}(\infty,x,y)f(i)p(j,i,2),
\end{align*}
where we recall that $p(x,y,2)$ denotes the probability of a price increase if there are $x$ orders at the bid, $y$ orders at the ask, and a spread of $2$. The probability $p(x,y,2)$ can be computed according to (\ref{pup:z2}), while $u_{(j,0)}(\infty,x,y)$ can readily be found from (\ref{uexplicit}) by making $t\to\infty$. It is worth mentioning that the case $z=1$ is arguably the most important in practice since, as empirically observed in several studies, the spread spends a great deal of time at level $1$.

{Next, let} $\mathscr{A}_{B_j}:=\{(1,j),(2,j),\ldots,(N^*,j)\}$. Now, for $z=2$,
\begin{align*}
\hat{p}(x,y,2)&=\sum\limits_{i=1}^{N^*}\sum\limits_{j=1}^{N^*}\Px[\varsigma(x,y)\leq {L},Y_{\varsigma(x,y)}(x,y)=(j,0), H=i, \tilde\varsigma(j,i)\leq {\tilde{L}}, \tilde{Y}_{\tilde\varsigma(j,i)}(j,i)\in \mathscr{A}_{B}]\\ 
			&\quad+ \sum\limits_{i=1}^{N^*}\sum\limits_{j=1}^{N^*}\Px[\varsigma(x,y)\leq {L}, Y_{\varsigma(x,y)}(x,y)=(j,0), H=i, \tilde{\varsigma}(j,i)\geq \tilde{L}, {\tilde{L}=\tilde{L}^{b}}]\\
			&\quad+ \sum\limits_{i=1}^{N^*}\sum\limits_{j=1}^{N^*}\Px[{L}<\varsigma(x,y),{L=L^{b}}, Y_{{L}}(x,y)\in \mathscr{A}_{B_j}, H=i,\tilde{Y}_{\tilde{\varsigma}(i,j)}(i,j)\in \mathscr{A}_{B}].
\end{align*}
{Hence, using that 
$\Px\left[L\leq\varsigma(x,y),L=L^{a}, Y_{a}(x,y)\in \mathscr{A}_{B_j}\right]=\Px\left[L\leq\varsigma(x,y),L=L^{b}, Y_{L}(x,y)\in \mathscr{A}_{B_j}\right]$, we can write
\begin{align*} 
	\hat{p}(x,y,2)&=\sum\limits_{i=1}^{N^*}\sum\limits_{j=1}^{N^*}f(i)\left\{\left( 2\alpha \int_0^\infty u_{(j,0)}(t,x,y)e^{-2\alpha t} dt\right) p(j,i,2)+ \frac{1}{2}\Px\left[{L}<\varsigma(x,y), Y_{{L}}(x,y)\in \mathscr{A}_{B_j}\right]p(i,j,1)\right\}.
\end{align*}
The} probability $p(x,y,1)$ can be computed according to (\ref{pup:z1}), while $2\alpha \int_0^\infty u_{(j,0)}(t,x,y)e^{-2\alpha t} dt$ can readily be found from (\ref{uexplicit}). The problem of computing {$\Px\left[L\leq\varsigma(x,y), Y_{L}(x,y)\in \mathscr{A}_{B_j}\right]$} is analyzed below. 
Before that, let us note that, using similar arguments,
\begin{align*}
\hat{p}(x,y,3)
			&=\sum\limits_{i=1}^{N^*}\sum\limits_{j=1}^{N^*} f(i)\left\{\left( 2\alpha \int_0^\infty u_{(j,0)}(t,x,y)e^{-2\alpha t} dt\right)p(j,i,2)+ \frac{1}{2}\Px\left[{L}<\varsigma(x,y), Y_{{L}}(x,y)\in \mathscr{A}_{B_j}\right]p(i,j,2)\right\}.
\end{align*}
A similar identity holds for $\hat{p}(x,y,z)$ with $z\geq{}4$. Therefore, the only {remaining step} is the computation of $\Px\left[L\leq\varsigma(x,y), Y_{L}(x,y)\in \mathscr{A}_{B_j}\right]$.
This can be done by first computing $v_{{j}}(t,x,y):=\Px[t\less \varsigma(x,y), Y_{t}(x,y)\in \mathscr{A}_{B_j} ]$ {using} similar arguments to those used in Proposition \ref{lemmausol}. More concretely, {it turns out that} $v_{j}(t,x,y)$ solves the initial value problem:}
\begin{equation}\label{PDE:prob1}
\left\{\begin{array}{rcl}
{\left.\left(-\frac{\partial}{\partial t} + \mathscr{L}\right) v_{j}(t,x,y)\right|_{t=T-r}=0} & \text{for} & 0\leq r \leq T,  
(x,y)\in\{1,2,\ldots,{N^*}\}^2, \\
v_{j}(T-r,x,y)=0 & \text{for} & 0\leq r\leq T,\, {(x,y)\in \mathscr{A}},\\
v_{j}(0,x,y)=\indicator{(x,y)\in \mathscr{A}_{B_j}} & \text{for} & (x,y)\in\{0,1,2,\ldots,{N^*}\}^2. \end{array}\right.
\end{equation}

\section{Numerical Examples}\label{Implementation}
{The purpose of this section is twofold. First, we analyze numerically the convergence of the midprice process towards  its diffusive limit process as established in Theorem \ref{Thm:finalFCLT}. Second, we compute some of the quantities of interest described in Section \ref{CmpImpQnt} and numerically study their behaviors under both our assumptions and those in \cite{CL2012}. For the first problem, {we develop} an efficient} simulation scheme {for the price process dynamics}, which is much more efficient than the {direct simulation of all the LOB events (i.e., limit, market, and cancellation orders)}. 

Recall that for the model {introduced} in Section \ref{MainModel}, the input parameters are the rates {$\lambda$, $\mu$, $\theta$, and $\alpha$. The first three parameters refer to the arrival rates of limit orders, market orders, and cancellation, respectively,} while $\alpha$ is the rate at which a new set of limit orders arrive in-between the bid-ask spread. Also, we need {the distributions $f^{b}$ and $f^{a}$  for the sizes of queues at the best bid and ask price, respectively, after the best bid and ask {price changes}. For simplicity, we set $f:=f^{a}=f^{b}$ and recall that we are assuming that $f^{a},f^{b}$ are supported on the finite set $\{1,\dots,{N^{*}}\}$.}

{For the subsequent numerical examples, we shall use the empirically estimated intensities described in Table \ref{table:parameters} below, which are {borrowed} from  \cite{CL2012} (see Table 3 therein). The time {units in the sequel are in seconds}. The maximum queue size ${N^{*}}$ is assumed to be $10$, with each unit representing a batch of 100 shares. {Unless otherwise specified, the} initial level I queue's configuration are set to be $(x,y)=(5,5)$, while the initial spread is {$\zeta_0=4$}. The distribution $f$ is taken to be uniformly distributed in $\{1,\dots,N^{*}\}$.  {Finally}, two different choices of $\alpha$ are {considered:} $\alpha=\upsilon+1$ and $\alpha=2\upsilon$.} 

\begin{center}
\begin{tabular}{|c|c|c|}\hline
 Stock & $\lambda$ & $\upsilon:=\mu+\theta$\\ \hline\hline
 Citigroup & 2204 & 2331\\ \hline
 General Electric & 317 & 325 \\ \hline
 General Motors & 102 & 104 \\ \hline
\end{tabular}
\captionof{table}{Estimates for the {intensities} of limit orders and market orders+cancellations, in number of batches per second (each batch representing 100 shares) on June 26th, 2008, {as reported in \cite{CL2012}}.}
\label{table:parameters}
\end{center}

\subsection{Simulation and Convergence Assessment}\label{Sect:SimSecondModel}

The most natural {(and naive)} way to simulate the price dynamics {would consist} of generating all {the LOB events or, equivalently, all the} Poisson {arrival times of orders (limit, market, and cancellations),} until {the time at which either the bid or ask queue gets depleted and there is consequently a price change. We would then reset the queue size at the side that got depleted and continue this process.} Unfortunately, this {procedure} is computationally intensive and not suitable to study the {coarse-grain behavior of the price process, especially for the purpose of Monte Carlo analysis where we require a large number of simulations}. Instead, we propose a more efficient method, in which we {directly simulate the random vector $(\varsigma(x,y),{Y_{\varsigma(x,y)}(x,y))}$, without simulating the events leading to it. This in turn would allow us to obtain directly the time at which the level I of the order book gets depleted (or equivalently, the time of a price change) and the amount of outstanding limit orders at the opposite side of the book. To simulate $(\varsigma(x,y),{Y_{\varsigma(x,y)}(x,y))}$, we take advantage of the representation for their joint probability given by Eq.~(\ref{uexplicit}). This representation has several advantages  since its computation requires to find the eigenvalues $\{\xi_k\}$ and eigenfunction  $\{f_k(x,y)\}$, \emph{only once},} regardless of $t$ and $\ax$.

By Proposition \ref{Ntnassympthm2} and Theorem \ref{Thm:finalFCLT}, we have
\begin{align}\label{RWWTC}
{\Ex\left(\frac{s_{{t}}}{t}\right)\stackrel{t\rightarrow\infty}{\rightarrow}\frac{\nu(h)}{\Ex_\pi(\tau_{1})},\qquad 
{\frac{{\rm Var}\left(s_{{t}}\right)}{t}}\stackrel{t\rightarrow\infty}{\rightarrow}\gamma^2(h),\qquad 
\Ex\left(\frac{{t}}{N_{{t}}}\right)\stackrel{t\rightarrow\infty}{\rightarrow} \Ex_\pi(\tau_{1})}.
\end{align}
In the sequel, we shall study the {performance of the above asymptotic approximations for {``large"} $t$. Our {goal} is to assess how close the distribution of $s_{t}$ is to its diffusive approximation for some sampling time spans, $t$, commonly used in practice (say, 1 min and 5 min). To compute the expectations and variances appearing in (\ref{RWWTC}), we use a Monte Carlo method} with 200 simulations of the order book. The results are shown in {Table \ref{table:Citigroup}.}
As expected, the larger {are} the rates $\lambda$ and $\up$, the {smaller $\Ex_\pi(\tau_{1})$ gets} and, as a result, the larger the expected rate of return {$\Ex\left(s_{{t}}\right)/t$} becomes. We also observe that, in that case, there seems to be a significant increment in the volatility {$\sqrt{{\rm Var}\left(s_{{t}}\right)}$} of the asset price. 
{This  is due to the fact that} increasing $\lambda$ and $\upsilon$ simultaneously is equivalent to {speeding up} the dynamics of the process, which will necessarily result in higher variability.

{\footnotesize 
\begin{center}
\begin{tabular}{|c | c c c c| c c c c| c c c c|} \hline
& \multicolumn{4}{|c|}{Scenario 1: $\lambda=2204$, $\up=2331$}&
\multicolumn{4}{|c|}{Scenario 2: {$\lambda=317$, $\up=325$}} &  
\multicolumn{4}{|c|}{Scenario 3: $\lambda=102$, $\up=104$}\\
\hline
\multirow{2}{*}{Case}& 
\multicolumn{2}{c}{$\alpha=2332$}& \multicolumn{2}{c|}{$\alpha=4662$}& \multicolumn{2}{c}{$\alpha=2332$}& \multicolumn{2}{c|}{$\alpha=4662$} & \multicolumn{2}{c}{$\alpha=2332$}& \multicolumn{2}{c|}{$\alpha=4662$} \\ %\cline{2-12}
 & $t=60$ & $t=300$ & $t=60$ & $t=300$  & $t=60$ & $t=300$ & $t=60$ & $t=300$ & $t=60$ & $t=300$ & $t=60$ & $t=300$  \\ \hline
$\Ex[s_{t}]/t$  & -7.02 & -6.44 & -6.69 & -6.55 & -3.50 & -3.55 & -3.35 & -3.61 & -1.57 & -1.74 & -1.78 & -1.62 \\ 
${\rm Var}[s_{t}]/t$ &  240 & 238 & 320 & 322 &  161 & 114 & 113 &  143 & 50 & 45 & 60 & 58\\ 
$\Ex[t/N_{t}]$ & $\frac{3.9}{1000}$ &  $\frac{3.9}{1000}$ &$\frac{3.9}{1000}$ & $\frac{3.9}{1000}$ & $\frac{8.6}{1000}$ &  $\frac{8.6}{1000}$ & $\frac{8.4}{1000}$ & $\frac{8.4}{1000}$ &  $\frac{19.3}{1000}$ &  $\frac{19.3}{1000}$ & $\frac{18.8}{1000}$ &  $\frac{18.8}{1000}$ \\[1ex] \hline
\end{tabular}
\captionof{table}{{Estimates of the expected return $\Ex[s_{t}]/t$, normalized variance  ${\rm Var}(s_{t})/t$, and expected time $\Ex_{\pi}(\tau)$ between price changes.}}
\label{table:Citigroup}
\end{center}
}

Next, we {turn our attention to the behavior of the} spread. Based again on $200$ simulation and {an initial spread of} 4, Table \ref{tablePercentage1} shows the percentage of the time that the spread spends at each state {during the time interval $[0,300\, {\rm sec}]$} for the different values of $\lambda$, {$\up$,} and $\alpha$. As shown {therein}, the larger {are} the rates $\lambda$ and {$\upsilon$}, the longer time the spread spends in {one} tick. This is due to the fact that the larger these rates {are}, the quicker the spread change and, by the choice of $\alpha$, the quicker it will {shrink} to 1. {More importantly, these results show that, when $\alpha/\upsilon$ is large enough, our model can closely replicate the stylized empirical behavior of the spread as illustrated, for instance, in \cite{CL2012} (see Table 2 therein).} 

Finally, {Figures \ref{fig:2332}-\ref{fig:210}} compare the empirical density of $s_{t}$, based on $200$ simulations, to a Gaussian density with mean and variance set equal to the respective sample mean and variance of the 200 replicas of $s_{t}$. {We do this for $t=1$min and $t=5$min, which {are} commonly used as sampling frequencies of many statistical estimation methods.} The empirical density  is obtained using the R function ``density", which computes a kernel density estimate\footnote{We use the default {parameter settings for the kernel and bandwidth given by R, which are respectively given according to a Gaussian kernel and the Silverman's ``rule of thumb" \cite[Eq. (3.31)]{Silverman}}.}. For sake of space, we only show the graphs corresponding to $\alpha=\upsilon+1$ (there is no significant changes when $\alpha=2\upsilon$). As seen in the graphs, the distribution of $s_{t}$ is {relatively} well approximated by a Normal distribution for {these two values of $t$}. 

\begin{center}
\captionof{table}{Distribution of the time spent by the spread for different values of $\lambda,\up$ and $\alpha$ {during $[0,300]$}.}  
\begin{tabular}{c c c c c c} \hline\hline 
\multicolumn{2}{c}{Case} & 1 Tick  & 2 Ticks & 3 Ticks & 4+ Ticks\\ [0.5ex] \hline 
\multirow{3}{*}{$\alpha=\up+1$}  & $\lambda=2204$, $\up=2331$    &  0.97248&	0.02716	& 0.00035	& 0.00001\\ 
																 	 	 & $\lambda=317$, $\up=325$ & 0.906891 &	0.088491 & 0.004135 &	0.000564\\ 
																 		 & $\lambda=102$, $\up=104$ & 0.86881 & 0.12084&0.008868&0.001482\\[1ex] \hline
\multirow{3}{*}{$\alpha=2\up$}    & $\lambda=2204$, $\up=2331$  & 0.98627 & 0.01365 & 0.00007 & 0.00001 \\ 
															    	 & $\lambda=317$, $\up=325$ & 0.95327 & 0.045697 & 0.000956 & 0.000077\\ 
														    		 & $\lambda=102$, $\up=104$ & 0.94383 &0.054443 &0.001579 &0.000148\\[1ex] \hline 
\end{tabular} 
\label{tablePercentage1}
\end{center}

\begin{figure}[H]
\centering
\begin{minipage}{.5\textwidth}
  \centering
  \includegraphics[width=0.8\linewidth]{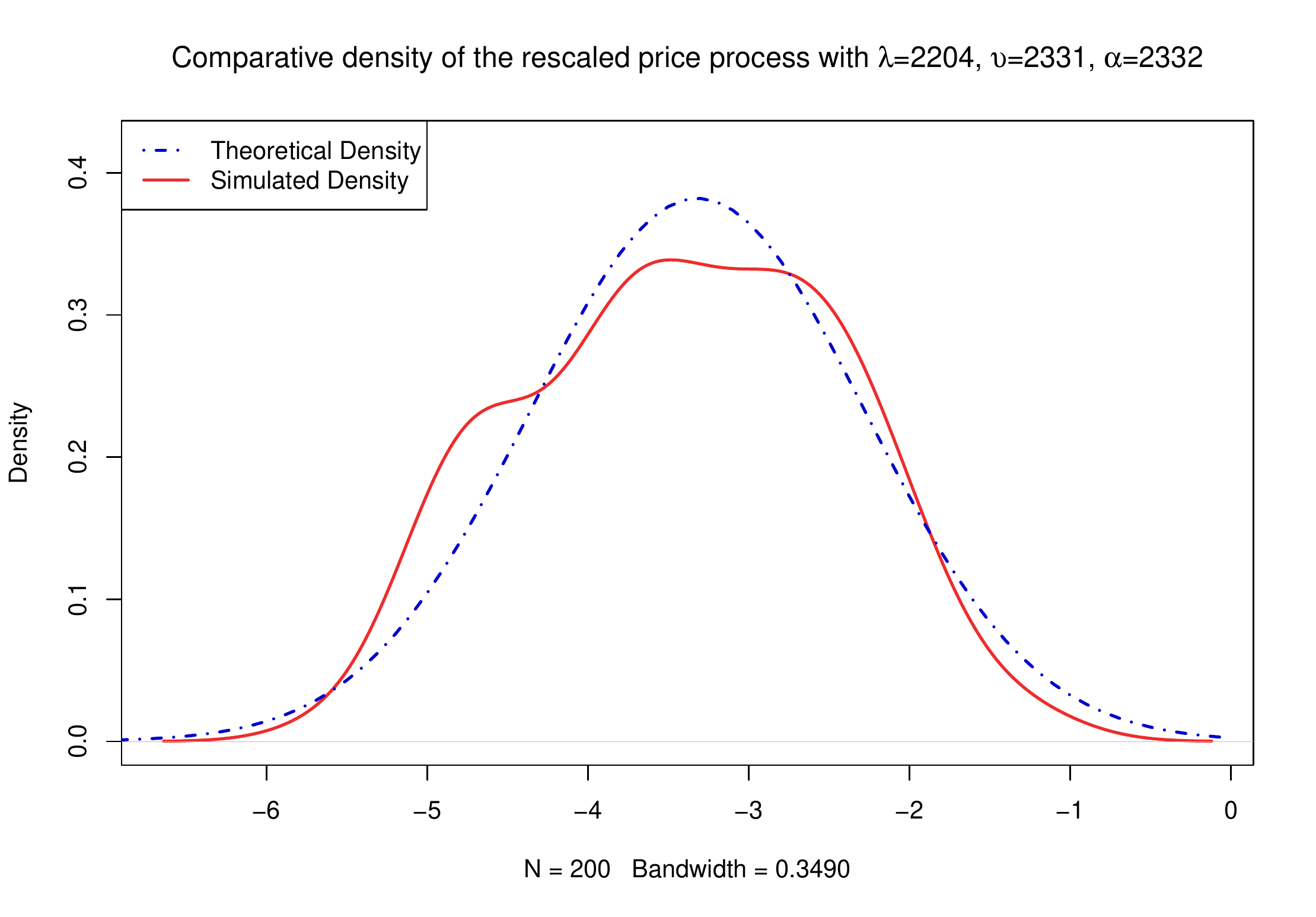}
\end{minipage}%
\begin{minipage}{.5\textwidth}
  \centering
  \includegraphics[width=0.8\linewidth]{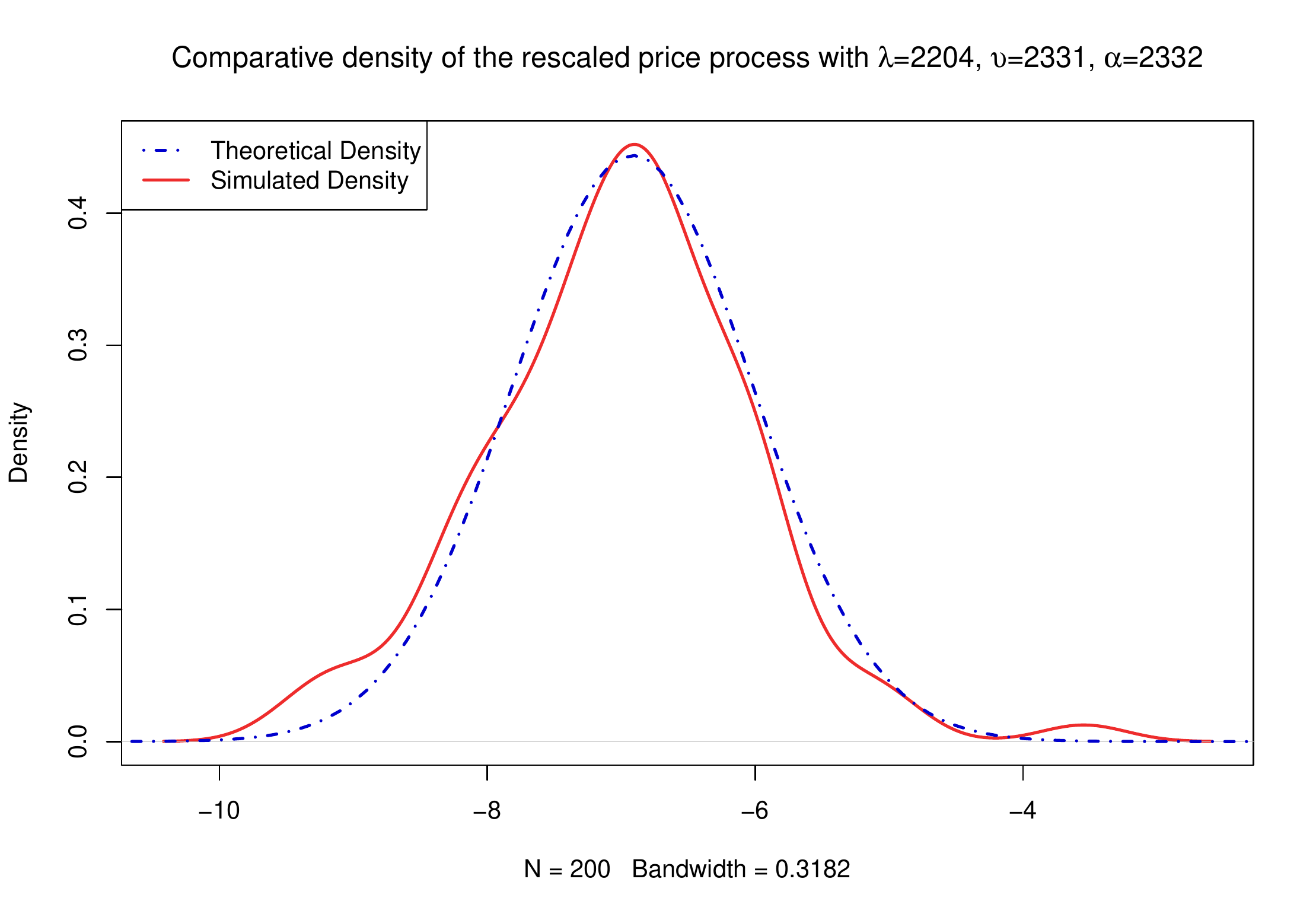}
\end{minipage}
  \captionof{figure}{{Comparison of the empirical density of the price $s_{t}$ to a Gaussian density}, when $\lambda=2204$, {$\upsilon=2331$,} and $\alpha=2332$. The time horizon chosen is $t=60s$ (left panel) and $t=300s$ (right panel).}
    \label{fig:2332}
\end{figure}

\begin{figure}[H]
\centering
\begin{minipage}{.5\textwidth}
  \centering
  \includegraphics[width=0.8\linewidth]{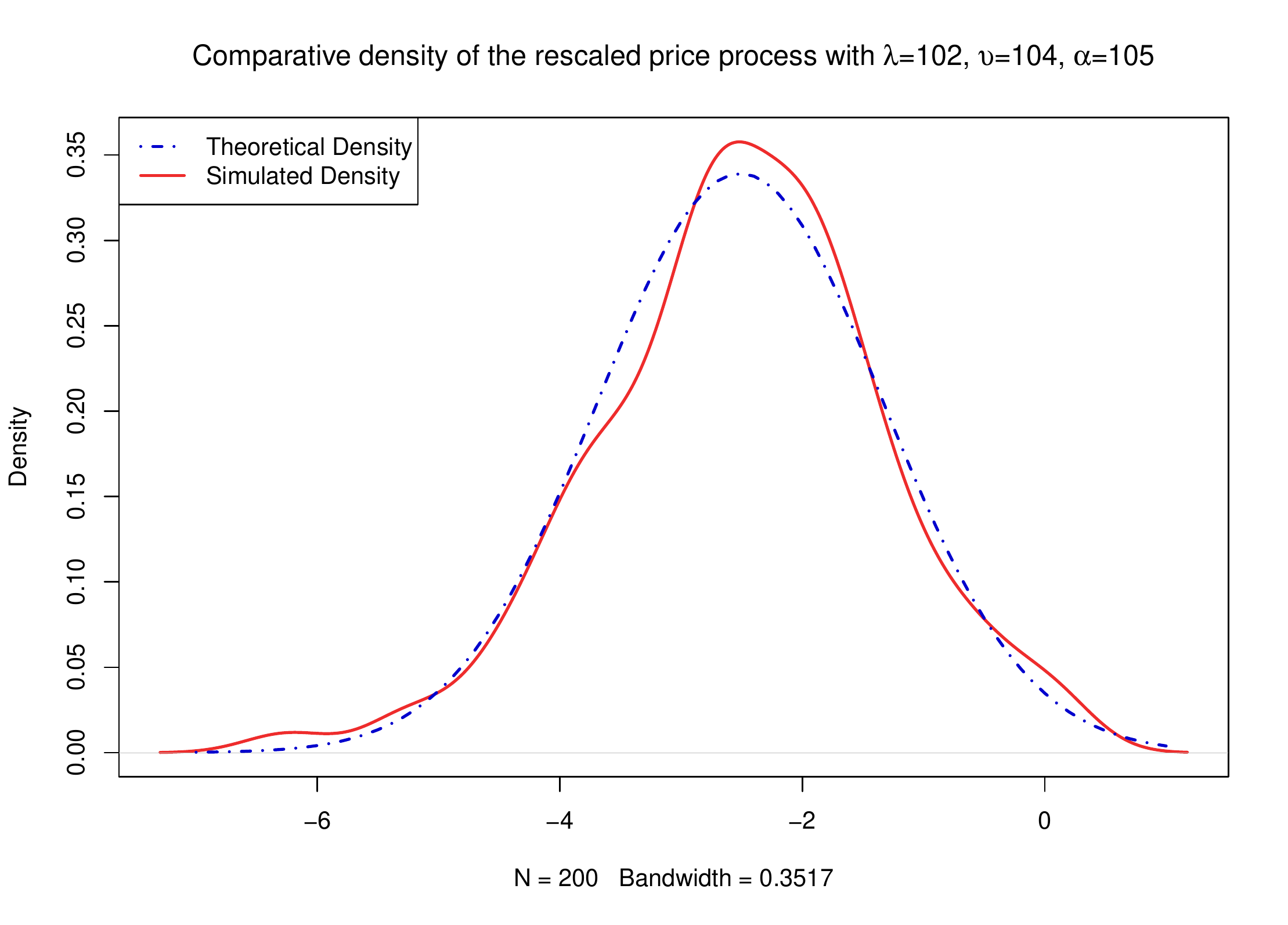}
\end{minipage}%
\begin{minipage}{.5\textwidth}
  \centering
  \includegraphics[width=0.8\linewidth]{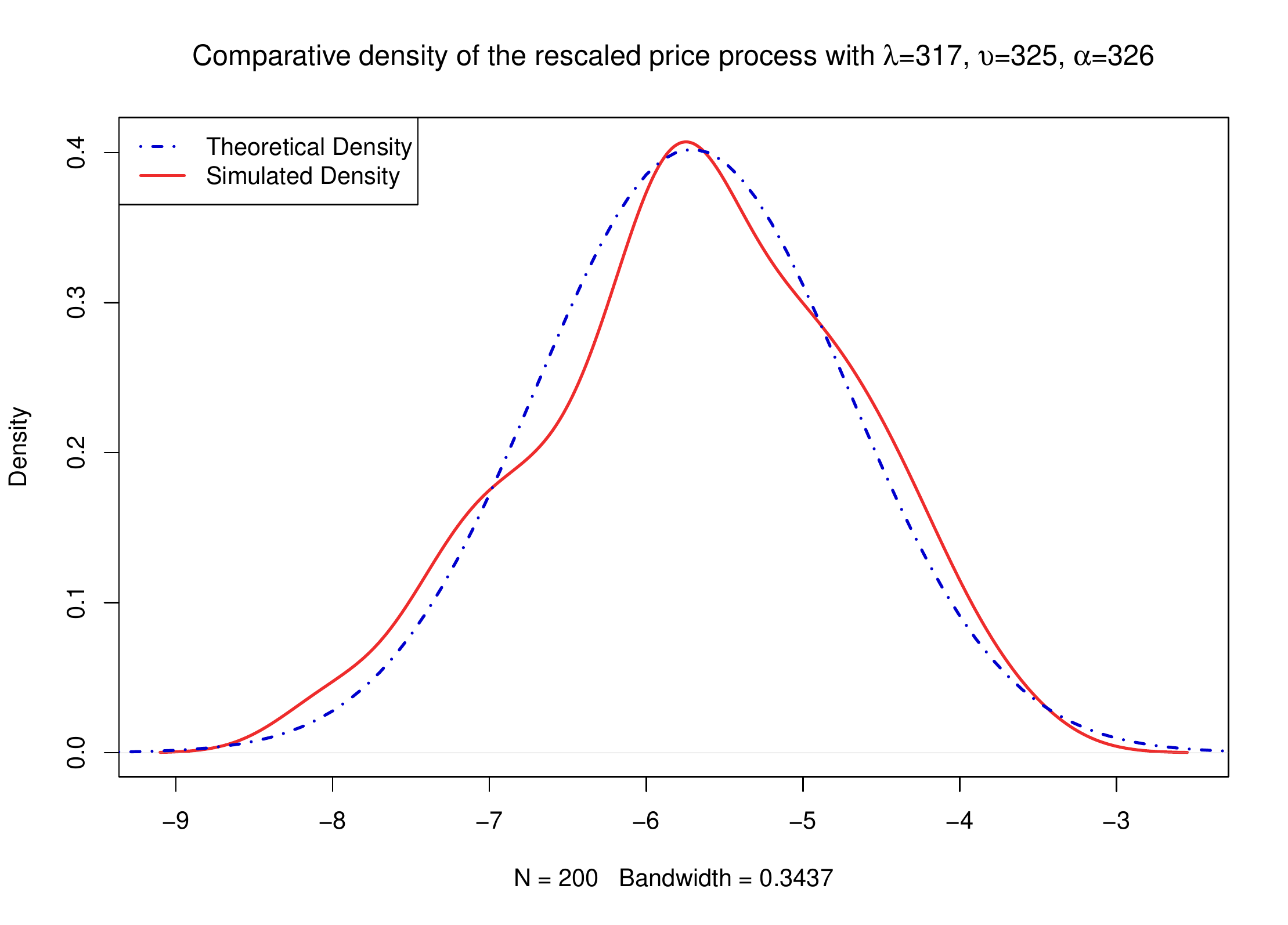}
\end{minipage}
  \captionof{figure}{{Comparison of the empirical density of the price $s_{t}$ to a Gaussian density}, when $\lambda=317$, {$\upsilon=325$,} and $\alpha=326$. The time horizon chosen is $t=60s$ (left panel) and $t=300s$ (right panel).}
    \label{fig:326}
\end{figure}

\begin{figure}[H]
\centering
\begin{minipage}{.5\textwidth}
  \centering
  \includegraphics[width=0.8\linewidth]{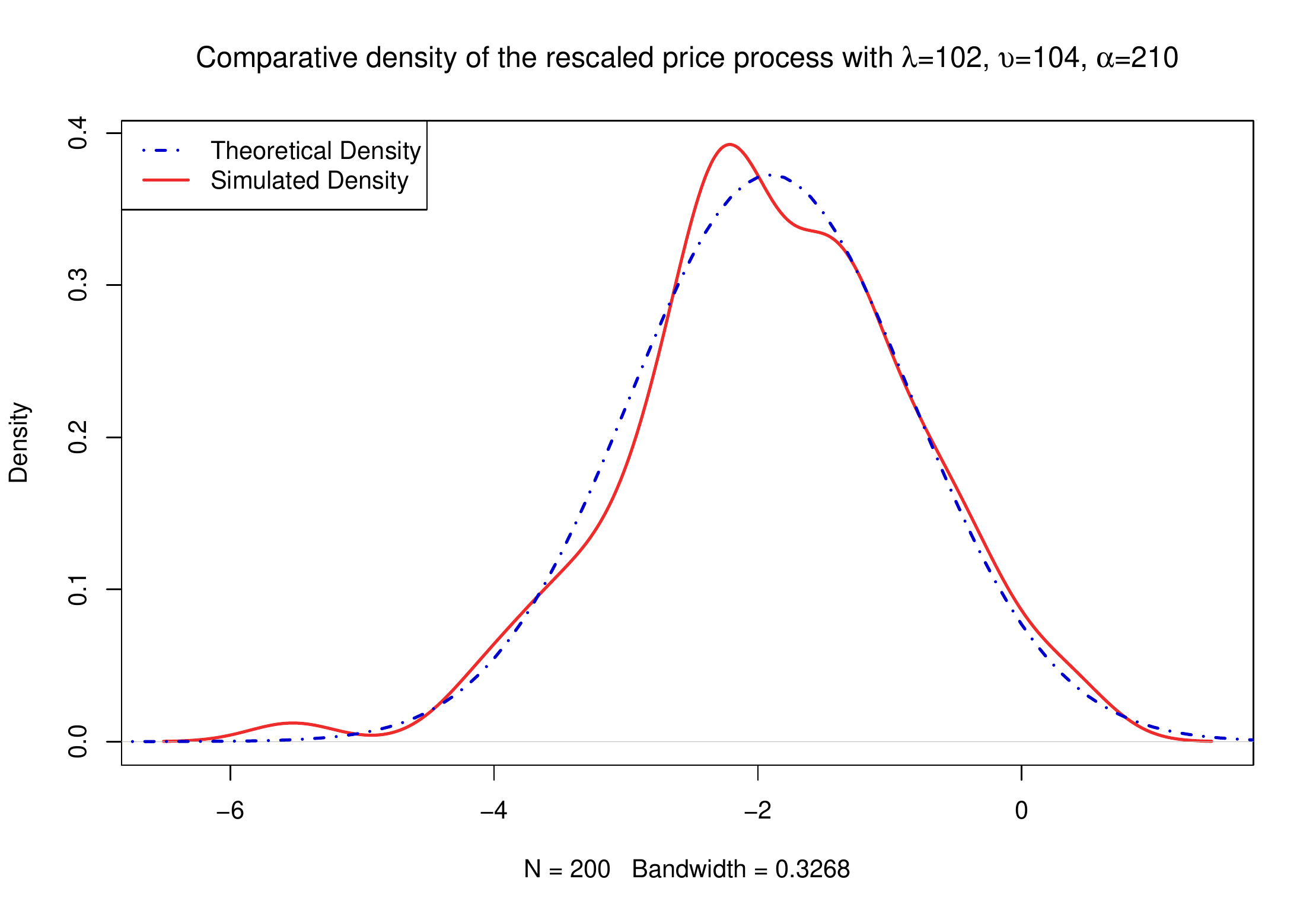}
\end{minipage}%
\begin{minipage}{.5\textwidth}
  \centering
  \includegraphics[width=0.8\linewidth]{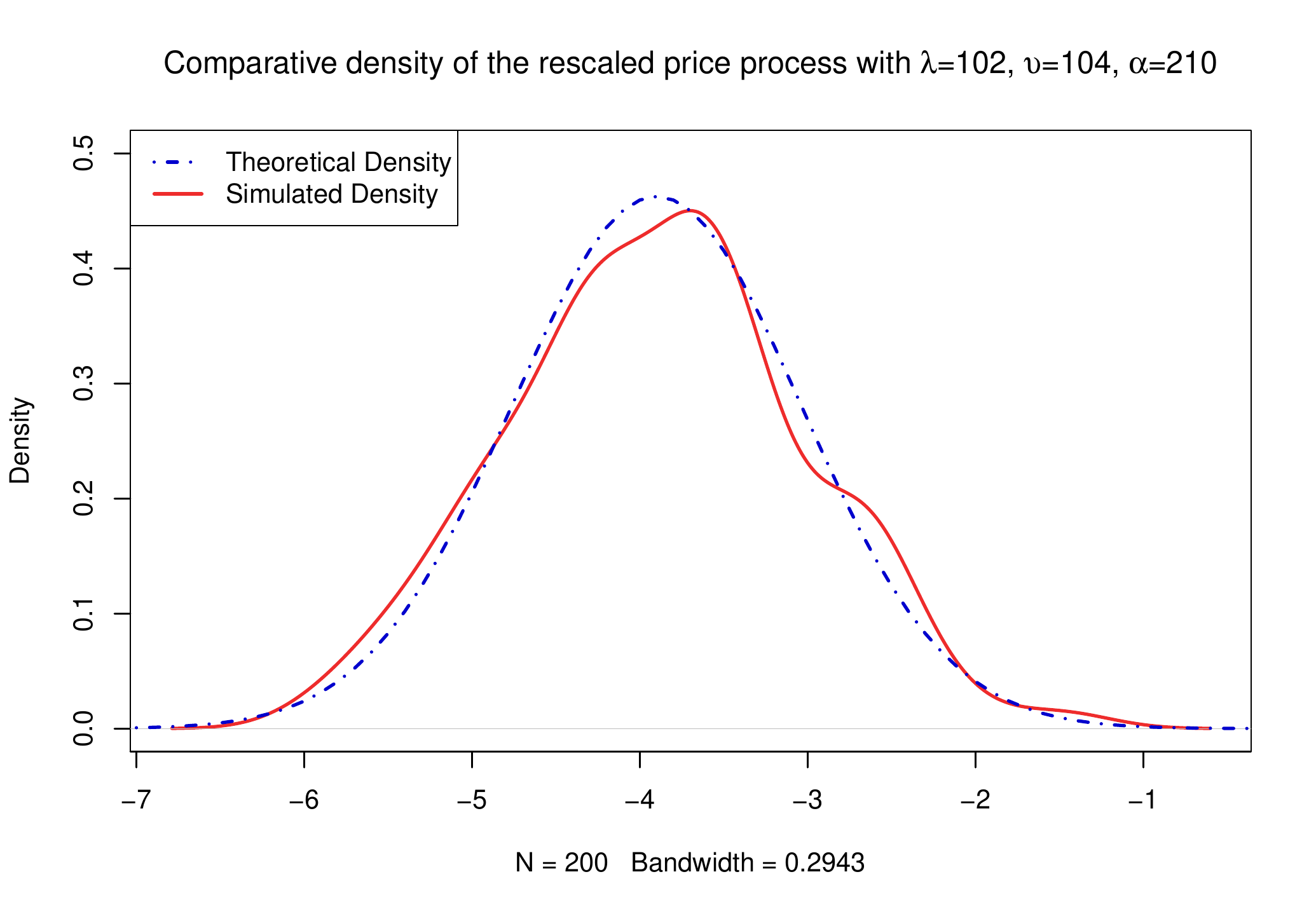}
\end{minipage}
  \captionof{figure}{{Comparison of the empirical density of the price $s_{t}$ to a Gaussian density}, when $\lambda=102$, {$\upsilon=104$,} and $\alpha=208$. The time horizon chosen is $t=60s$ (left panel) and $t=300s$ (right panel).}
    \label{fig:210}
\end{figure}

\subsection{Evaluation of some quantities of interest} 

To understand the impact of the assumptions made in the model {and draw some further comparisons to} the model presented in \cite{CL2012}, in this section, {we numerically compute some of the quantities of interest introduced in Section \ref{CmpImpQnt}}. 

{We first consider the distribution of the time span between price changes}. This distribution was {compute in Proposition 1 of \cite{CL2012}, under the assumptions therein. The survival function { $P(\tau\geq{}t)$} was also} plotted in Figure 4 therein with $\lambda=12$, $\mu+\theta=13$, $x_0^a=5$, and $x_0^b=4$ as the input parameters. In the left panel of Figure \ref{fig:Survival} below, this survival probability distribution is {reproduced and compared with the distributions obtained by the method introduced in Section \ref{sec:Distr:time} (see Eq.(\ref{Distribution:tau:z1})-(\ref{Distribution:tau:z2})), conditional on a spread of $1$, for different values of $N^*$}. {The right panel of Figure \ref{fig:Survival} also} depicts the densities of the time for the next price change.
{As it can be seen from the} plots, for values {of $N^*$ close to 5, the density is more concentrated around 0, which is natural since the queue sizes cannot increase more than $N^*$, causing this time to occur faster. Bu, as it is expected, the survival and density functions under our model converge to those  of \cite{CL2012} when $N^{*}$ increases.}

\begin{figure}[H]
\centering
  \includegraphics[height=7.0cm, width=8.5cm]{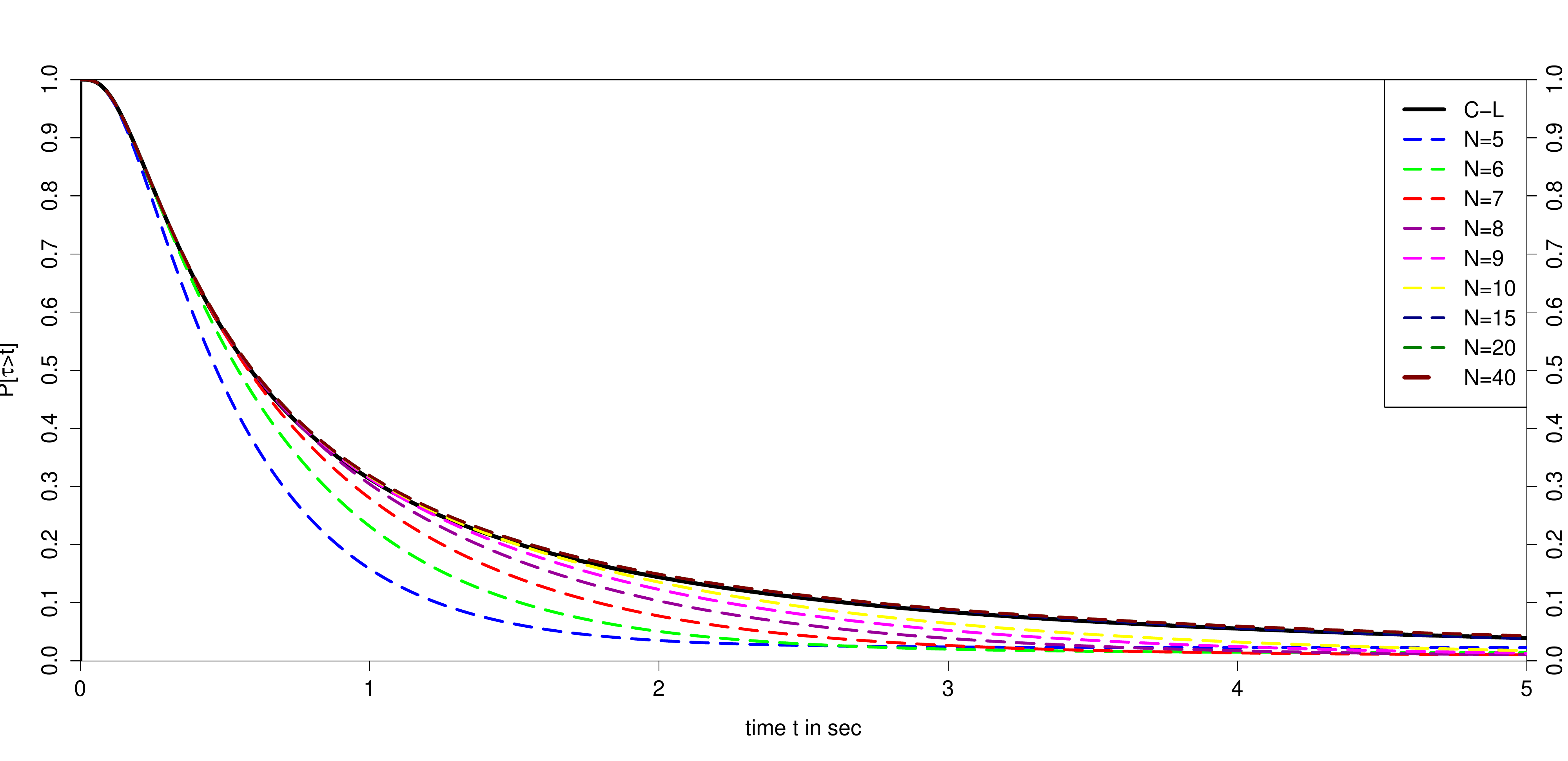}\hspace{0.2 cm}
  \includegraphics[height=7.0cm, width=8.5cm]{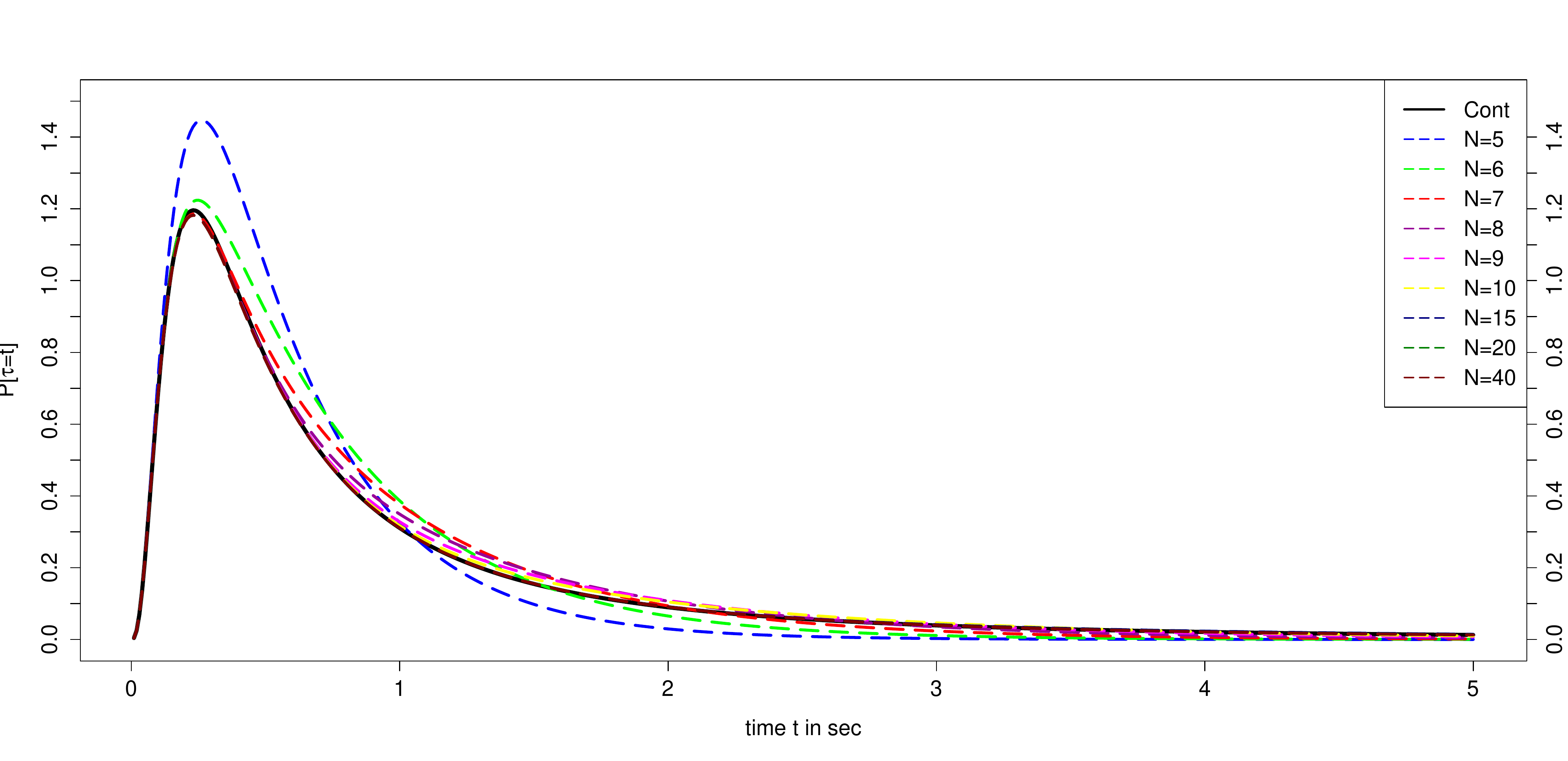}
\caption{Comparison of the survival function {and density} of the time for the first price change to occur between the model presented in \cite{CL2012} and different values of $N^*$ when the spread is set to be 1. {Parameter choices: $\lambda=12$, $\mu+\theta=13$, $x_0^a=5$, and $x_0^b=4$.}}
    \label{fig:Survival}
\end{figure}

The distribution {of the time for the next price to occur, when the spread is 2, is not plotted, because this} is very similar to the one of an exponential random variable with parameter $2\alpha$. This is due to the fact that the recurrence condition $\alpha\geq\mu+\theta$ implies that it is {far more probable that a price change occurs due to the arrival of a new set of limit orders within the spread than the depletion of a level I queue.}

Next, {we compare the probability of a price increase for different values of $x_0^a$, $x_0^b$, and $N^*$, when the spread is set to be 1. \cite{CL2012} provides a formula (see Proposition 3), under the assumptions therein, but, unfortunately, this formula is difficult to implement in the asymmetric order flow case. In contrast, the method proposed in Subsection \ref{sec:PriceIncrease2} is more efficient, since all of the quantities therein rely on the spectral decomposition of the discrete Laplacian (\ref{vGenerator}), which, once $N^*$ is fixed, just has to be done once. Figures 
%\ref{fig:Prop:Incr10}, 
\ref{fig:Prop:Incr30} graphs the probability of price increase as a function of $x_{0}^{a}$, for fixed $x_0^b=30$ (left panel) and $x_{0}^{b}=50$ (right panel). By symmetry, we would have the same graph against $x_{0}^{b}$ for fixed values of $x_{0}^{a}$. Notice that, in the first case, when $x_0^b=30$, the probability of a price increase does not varies significantly with $N^{*}$, for a most of the values of $x_{0}^{a}$. It is only when $x_0^a$ becomes close to the value of $N^*$ that some discrepancies start to show. On the other hand, when $x_0^b=50$, which is a closer value for $N^*$, the probability significantly varies with $N^*$ regardless of the value of $x_0^a$. The dashed lines therein show that, regardless of $N^*$, the probability of a price increase is always 0.5,  as it should be, when $x_0^a=x_0^b$.}

\begin{figure}[H]
\centering
  \includegraphics[height=7.5cm, width=8.8cm]{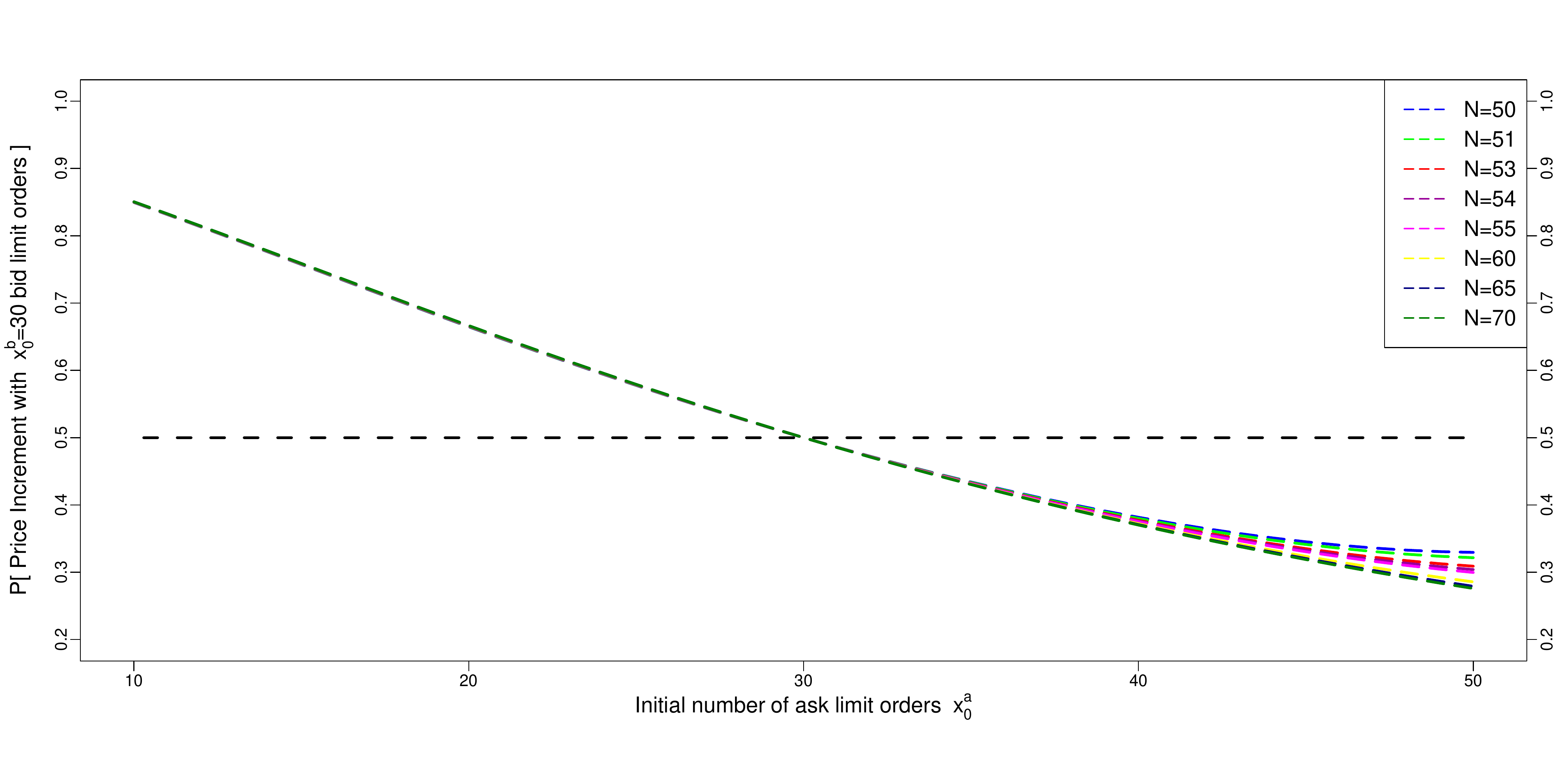}\hspace{0.3 cm}
  \includegraphics[height=7.5cm, width=8.8cm]{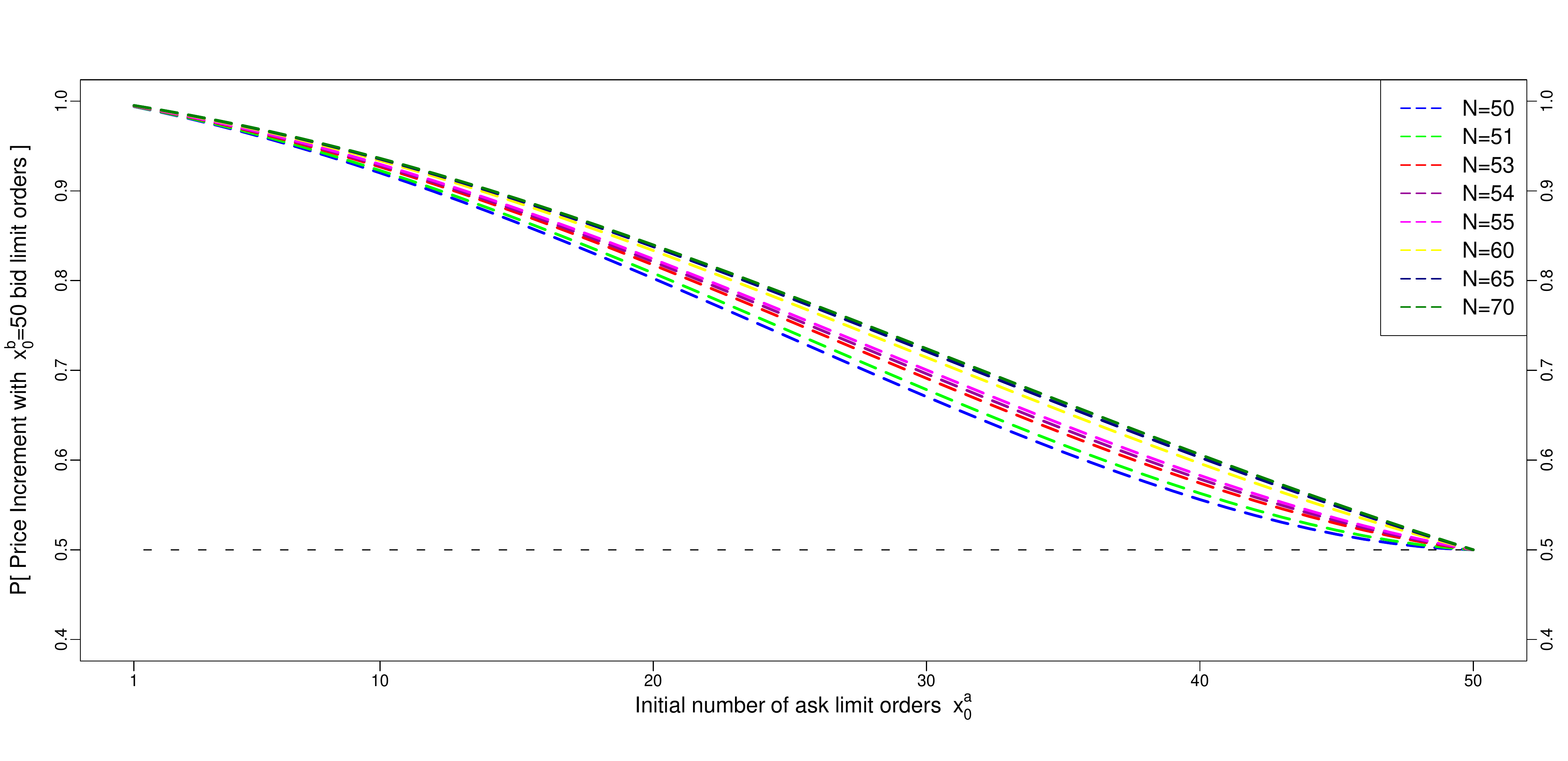}
  \caption{{Comparison of the probability of a price increase as a function of $x_{0}^{a}$ for different values of $N^*$. The number of bid orders is fixed at 30 (left panel) and 50 (right panel).}}
    \label{fig:Prop:Incr30}
\end{figure}

\section{Conclusions}
In this paper, {a new Markovian limit order book model is presented, which allows to incorporate a non-constant spread and to keep} the information about the outstanding orders after {a level I queue gets depleted at either side of the LOB}. Although the general rules governing the order book in {the latter} setting create a more complex dynamics, a novel efficient {method was developed to analyze several features of the LOB model of relevance in high-frequency trading}.

Our main result characterizes the coarse-grain behavior of the midprice process in terms of a Brownian motion with drift. This was made possible by expressing the price changes in terms of a suitable Markov chain. To this end, two key assumptions were needed:
the boundedness of the {queue sizes} at every moment and a sufficiently high arrival rate of new orders in between the spread {compared to the intensity of arrivals of market order/cancellations}. The latter condition is also intuitive since it prevents the spread to grow indefinitely with positive probability. {These} two conditions provide a tractable framework, without {lost of realism}, for the LOB, and become relevant when studying the diffusive behavior of the price process for this model.

{It} is known that markets exhibit relatively large price shifts in {a} small time period {and, thus, the incorporation of these ``jumps" into an order book model is appealing. A natural approach to address this problem may be the introduction of more levels in the order book, governed by similar rules to those imposed in our second proposed model. {The approach presented in this work is also} expected to be applicable for such models.} 

\medskip
\noindent
{\textbf{Acknowledgments:} 
The first author's research was supported in part by the NSF Grant DMS-1149692. The authors are grateful to two anonymous referees and both the associate editor and the editor-in-chief for their constructive and insightful comments that greatly helped to improve the paper.}

\appendix

\section{{{Additional} Proofs}}\label{ApndB}

\begin{proof}[{\textbf{Proof of Lemma \ref{phiharm}}}]
Throughout, {we set $Y_{t}(x,y):=(Y^{1}_{t}(x,y),Y^{2}_{t}(x,y)):=(Q_{t}^{b,0}(x),Q^{a,0}(y))$, $L^{a}:=L^{a}_{1}(2)$, $L^{b}:=L^{b}_{1}(2)$, and $L:=L^{a}\wedge L^{b}$.}
{Also, let $\hat{y}=(y_1,j\pm1,y_2,j)\in\Xi\backslash F$, where $y_1=(y_1^1,y_1^2)\in\Omega_{N^*}^2$, $y_2=(y_2^1,y_2^2)\in\Omega_{N^*}^2$ and $j\more1$. In that case,} 
\begin{align*}
{(P\varphi)(\hat{y})}
&=\sum\limits_{{z_2=(z_{2}^{1},z_{2}^{2})}\in\Omega_{N^*}^2}{P((y_1,j\pm1,y_2,j),(y_2,j,z_2,j-1))}\,\varphi((y_2,j,z_2,j-1))\\
		&\quad+ \sum\limits_{{z_2=(z_{2}^{1},z_{2}^{2})}\in\Omega_{N^*}^2}{P((y_1,j\pm1,y_2,j),(y_2,j,z_2,j+1))}\,\varphi((y_2,j,z_2,j+1)),
\end{align*}
{which,  using that $\varphi((y_{2},j,z_{2},k))=\varphi(k)$, can then be decomposed and simplified as follows:
%\small
\begin{align*}
{(P\varphi)(\hat{y})}
	&=\sum\limits_{z_2^1=1}^{\ N^*}\sum\limits_{z_2^2=1}^{\ N^*}\Px(L\less\varsigma(y_2), L=L^{a}, Y_{L}^{1}(y_2)=z_2^1)f^{a}(z_2^2){\varphi}(j-1)\\
				&\quad +\sum\limits_{z_2^1=1}^{\ N^*}\sum\limits_{z_2^2=1}^{\ N^*}\Px(L\less\varsigma(y_2), L=L^{b}, Y_{L}^{2}(y_2)=z_2^2){f}^{b}(z_2^1){\varphi}(j-1)\\
		&\quad+\sum\limits_{z_2^1=1}^{\ N^*}\sum\limits_{z_2^2=1}^{\ N^*}\Px(\varsigma(y_2)\less L, Y_{\varsigma(y_2)}(y_2)=(z_2^1,0))f^{a}(z_2^2){\varphi}(j+1)\\
		&\quad +\sum\limits_{z_2^1=1}^{\ N^*}\sum\limits_{z_2^2=1}^{\ N^*}\Px(\varsigma(y_2)\less L, Y_{\varsigma(y_2)}(y_2)=(0,z_2^2))f^{b}(z_2^1){\varphi}(j+1)\\
		&={\varphi}(j-1)\Px(L\less\varsigma(y_2))+{\varphi}(j+1)\Px(\varsigma(y_2)\less L).
\end{align*}
Since} {$\Px(\varsigma(y)\more t)\leq\Px(\varsigma(z)\more t)$} for $y =(y^1,y^2),\ z=(z^1,z^2)\in\Omega_{N^*}^2$ with $z^1\geq y^1$ and $z^2\geq y^2$, for any $y_2\in\Omega_{\N^*}^2$,
\[\Px(\varsigma((N^*,N^*))\more t)\geq\Px(\varsigma(y_2)\more t)\geq\Px(\varsigma((1,1))\more t).\]
 Thus, 
\begin{align*}
\sum\limits_{\hat{z}\in\Xi}{P(\hat{y},\hat{z})\varphi(\hat{z})}&={\varphi}(j-1)\Px(N\less\varsigma(y_2))+{\varphi}(j+1)\Px(\varsigma(y_2)\less N)\\
		&\leq{\varphi}(j-1)(1-p_{\mathbf{N^*}})+{\varphi}(j+1)p_{\mathbf{1}},
\end{align*}
From the previous expression, a sufficient condition for ${\varphi}$ to be super-harmonic, is to satisfy the {linear difference equation $p_{\mathbf{1}}{\varphi}(j+1)+(1-p_{\mathbf{N^*}}){\varphi}(j-1)={\varphi}(j)$, whose particular solution, satisfying the desired boundary conditions, is given by (\ref{phisupharm}).}
\end{proof}

\begin{proof}[{\textbf{Proof of Lemma \ref{alpha}}}] {Throughout, we set ${\bf 1}=(1,1)$ and ${\bf N^{*}}=(N^{*},N^{*})$.} First, we will prove that the {condition $\alpha\more \mu+\theta$ implies} an upper bound for $p_{\mathbf{1}}$. The independence of $\varsigma(\mathbf{1})$ and ${L}\sim \exp(2\alpha)$ {implies} that {$\Px({L}\more\varsigma(\mathbf{1}))=\int_0^\infty f_{\varsigma(\mathbf{1})}(t) e^{-2\alpha t} dt$,} where $f_{\varsigma(\mathbf{1})}(t)$ is the probability density functions of $\varsigma(\mathbf{1})$. Using integration by parts,
\begin{equation} \label{p1bound}
\Px({L}\more\varsigma(\mathbf{1}))=\int_0^\infty \frac{d}{dt}\left(-\Px(\varsigma(\mathbf{1})\geq t)\right) e^{-2\alpha t} dt=1-\int_0^\infty 2\alpha e^{-2\alpha t}\Px(\varsigma(\mathbf{1})\geq t) dt.
\end{equation}
{Let $E_{t}$ be the event that there is neither {a} cancellation nor {an} arrival of market orders before time $t$ at either side of the book.} Since {$\Px(\varsigma(\mathbf{1})\geq t)\more {\Px(E_{t})}=e^{-{2(\mut)t}}$,} by (\ref{p1bound}) and the assumption that {$\alpha>\mu+\theta$},
\[
\Px({L}\more\varsigma(\mathbf{1}))=1-\int_0^\infty 2\alpha e^{-2\alpha t}\Px(\varsigma(\mathbf{1})\geq t) dt\less 1-\frac{\alpha}{\alpha+\mut}\leq \frac{1}{2}.
\]
Thus, regardless of the sign of {$p_{{\bf 1}}(1-p_{\bf N^*})$}, since {$p_{{\bf N^*}}\less p_{\bf 1}\less\frac{1}{2}$}, {we have that} $\lim_{j\rightarrow\infty}{\varphi}(j)=\infty$.
\end{proof}

\begin{proof}[{\textbf{Proof of Lemma \ref{FntnessIntg}}}] 
	We apply Theorem \ref{LLNHarris}, for which we need to prove that $\pi(|f|)<\infty$ and $\pi(|g_t|)<\infty$. The latter {assertions} hold true if we {can} show that, for all $\hat{x}=(x_0,c_0,x_1,c_1)\in\Xi$, 
	\begin{equation}\label{expectau}
		\Ex(\tau_1\left.\right|\hat{x})	\leq C<\infty,
\end{equation}
for a constant $C$, since
	\begin{align*}
		\pi(|f|):=\sum_{\hat{x}\in\Xi}\pi(\hat{x})\Ex(\tau_1\left.\right|\hat{x})\leq C\sum_{\hat{x}\in\Xi}\pi(\hat{x})\less\infty,\qquad
		\pi(|g_t|):=\sum_{\hat{x}\in\Xi}\pi(\hat{x})\Px(\tau_1\more t\left.\right|\hat{x})\leq\sum_{\hat{x}\in\Xi}\pi(\hat{x})\less\infty.
\end{align*}
To show (\ref{expectau}), we first need some notation. {Let 
$\varsigma(x)$ be defined as in Lemma \ref{phiharm}.} 
Note that
\begin{equation}\label{taufinitemoment}
	{\Ex\left(\left.\tau_1\right|(\tilde{x}_0,\zeta_0\right)=(x_0,c_0))\leq\Ex\left(\left.\tau_1\right|(\tilde{x}_0,\zeta_0)=((N^*,N^*),1)\right)=\Ex\left(\varsigma((N^*,N^*))\right)\less\infty},
\end{equation}
where the last inequality holds, since $\varsigma((N^*,N^*))\leq\min(\varpi_1,\varpi_2)$, where $\varpi_i$, i=1,2, is the hitting time at 0 of a 1-dimensional birth and death process with birth rate $\lambda$ and death rate $\mu+\theta$ starting at $N^*$ (for which is known the expectation is finite) and $\varpi_1$ is independent of $\varpi_2$.
{Next, let $\mathcal{R}(x_{0},c_{0})=\{x_{1}\in{\Omega_{N^{*}}^{2}}:\Px((\tilde{x}_1,\zeta_1)=(x_{1},c_0\pm1)\left.\right|(\tilde{x}_0,\zeta_0)=(x_0,c_0))>0\}$ and let 
\begin{align*}
	r^{\pm}_{x_{1}}((x_0,c_0))&:=\Px((\tilde{x}_1,\zeta_1)=(x_{1},c_0\pm1)\left.\right|(\tilde{x}_0,\zeta_0)=(x_0,c_0)),\quad c_{0}>1,\\ 
	r_{x_{1}}((x_0,1))&:=\Px((\tilde{x}_1,\zeta_1)=(x_{1},2)\left.\right|(\tilde{x}_0,\zeta_0)=(x_0,1)), \\	r_{\min}(x_0)&:=\min\{r^{\pm}_{x_{1}}((x_0,2)):x_{1}\in\mathcal{R}(x_{0},2)\}\wedge\min\{r_{x_{1}}((x_0,1)):x_{1}\in\mathcal{R}(x_{0},1)\}.
\end{align*}
Since, for any} $c_0,c_1\more1$, $r^{\pm}_{x_{1}}((x_0,c_0))=r^{\pm}_{x_{1}}((x_0,c_1))$, it follows that {$0\less r_{\min}(x_0)\leq r^{\pm}_{x_{1}}((x_0,c_0))$ for all $c_0\in\{2,3,\ldots\}$ and $x_{1}\in\mathcal{R}(x_{0},c_{0})$}. Therefore, 
\[
	r_{\min}(x_0)\sum\limits_{{x_1\in\mathcal{R}(x_{0},c_{0})}}\Ex(\tau_1\left.\right|x_0,c_0,x_1,c_{0}\pm 1)\leq \Ex(\tau_1|(\tilde{x}_0,\zeta_0)=(x_0,c_0))<\Ex(\varsigma((N^*,N^*)))\less\infty.
\]
{This implies (\ref{expectau}), which in turn implies the result as explained above.}
\end{proof}	

\begin{proof}[{\textbf{Proof of Theorem \ref{FCLTh}}}]  Since the state space, $\Lambda$, is countable, every finite subset of the state space is an atom (e.g., see \cite[Chapter 5, pg 105]{Meyn} ) and, hence, we are able to construct explicitly the solution of the Poisson equation (\ref{poissoneqn}). Indeed, by Equation (17.38) in \cite{Meyn} and the discussion therein, for $C_1:=\{\oz=(x,c,u)\in\Lambda: x\in\Omega_{N^*}^2, c=1, u\in\{-1/2,1/2\}\}$, we have that
\begin{equation}\label{Eqn:Soln:poissonexplicit}
\hat{h}({\oz})=\Ex_{{\oz}}\left[\sum\limits_{{k=1}}^{\sigma_{C_1}} \bar{h}(V_k)\right],
\end{equation}
where $\sigma_{C_1}=\min\{n\geq0\;|\;V_n\in C_1\}$. Since for any ${\oz}\in\Lambda$, {$|h(\oz)|\leq 1/2$}, {$|\overline{h}(\oz)|\leq1$ and, thus, ${|\hat{h}(\oz)|}\leq \Ex_{\oz}(\sigma_{C_1})$}. Therefore, to {conclude that the invariance principle (\ref{InvTh2}) holds true, it suffices to} show that 
\begin{equation}\label{expsigma}
{\nu\left(\Ex^2_{{\cdot}}(\sigma_{C_1})\right)}:=\sum_{\oz\in\Lambda}\nu(\oz)\Ex_{\oz}^2(\sigma_{C_1})\less\infty.
\end{equation}
 Let $C_j=\{\oz=(x,c,u)\in\Lambda| x\in\Omega_{N^*}^2, c=j, u\in\{-1/2,1/2\}\}$. Each $C_j$ is finite and $\{C_j\}_{j\geq1}$ forms a partition of $\Lambda$. Clearly, for every $n$, if $V_n\in C_j$, with $j\geq2$, then $V_{n+1}\in\{C_{j-1},C_{j+1}\}$.  Moreover, with the notation of Lemma \ref{phiharm} and, {as prove in the proof of} Lemma \ref{alpha}, for any $\oz=(x,c,u)\in C_j$
\[
 	\Px(V_{n+1}\in C_{j+1}|V_n =\oz)=\Px\left(\varsigma(x)<L\right)\leq\Px(\varsigma(\mathbf{1})< {L})=p_{\mathbf{1}}
\]

Consider now a birth and death process $\tilde{V}_n\in\N$ with birth probability $p_{\mathbf{1}}$ and death probability $1-p_{\mathbf{1}}$, {and note that
\begin{align*}
\Px[V_{n+1}\in C_{j-1}\;|\;V_n=\oz]&\geq 1-p_{\mathbf{1}} = \Px[\tilde{V}_{n+1}=j-1\;|\;\tilde{V}_n=j].
\end{align*}
Denote} by $\sigma_1^{\tilde{V}}$ the first hitting time of {$\tilde{V}$ to the point $1$. That is,} $\sigma_{1}^{\tilde{V}}=\min\{n\more0\;|\;\tilde{V}_n=1\}$. 
Then, {since ${V}$ dies more frequently than $\tilde{V}$, for any $\oz\in C_j$,}
\begin{equation}\label{eqn:compar:sigma}
\Ex_{\oz}[\sigma_{C_1}]\leq\Ex_j[\sigma_1^{\tilde{V}}]= \frac{j-1}{1-2p_{\mathbf{N^*}}},
\end{equation}
where, in the last equality we use that $p_{\mathbf{1}}\less 1/2$ (see \cite{Sericola}[Section 3.1]).

The next step is to bound the terms $\nu(C_j)$ for $j\geq2$. To shorten notation, define $\Theta:=\Omega_{N^*}\times\{-1/2,1/2\}$.
Recall that if the spread is larger than one, the spread will widen {or shrink} right after every price change, whereas if the spread is 1, it will surely widen at the next step. Thus, for  any $\oy\in C_1$ and any $\oz\in C_{j}$ with $j\geq2$, {$P^{ext}[\oy,C_{2}]=1$ and $
P^{ext}[\oz,C_{j+1}]+P^{ext}[\oz,C_{j-1}]=1$.}
By the definition of a stationary {measure, 
\begin{align*}
\nu(C_1)= \int_{\Lambda}P^{ext}[\oz,C_1]\nu(d\bar{z})
= \sum\limits_{\oz\in C_{2}} P^{ext}[\oz,C_1]\nu(\oz)
= \sum\limits_{\oz\in C_{2}} (1-P^{ext}[\oz,C_3])\nu(\oz),				
\end{align*}
which} implies that
\begin{equation}\label{eqn:rel:nuC3}
\sum\limits_{\oz\in C_{2}} P^{ext}[\oz,C_3]\nu(\oz) = \nu(C_2)-\nu(C_1)
\end{equation}
{Analogously,
\begin{align*}%\label{eqn:defn:nuC2} \nonumber
\nu(C_{2})&= \int_{\Lambda}P^{ext}[\oz,C_{2}]\nu(d\bar{z})= \sum\limits_{\oz\in C_{1}} P^{ext}[\oz,C_{2}]\nu(\oz) + \sum\limits_{\oz\in C_{3}} P^{ext}[\oz,C_{2}]\nu(\oz)
%				&= \sum\limits_{\oz\in C_{1}} \nu(\oz) + \sum\limits_{\oz\in C_{3}}(1-P^{ext}[\oz,C_{4}])\nu(\oz)\\ \nonumber
				= \nu(C_1) + \sum\limits_{\oz\in C_{3}}(1-P^{ext}[\oz,C_{4}])\nu(\oz),
%				\\
%				&= \nu(C_1) + \nu(C_3) - \sum\limits_{\oz\in {C_{3}}}P^{ext}[\oz,C_{4}]\nu(\oz),
\end{align*}
which implies that}
\begin{equation}\label{eqn:rel:nuC4}
\sum\limits_{\oz\in {C_{3}}}P^{ext}[\oz,C_{4}]\nu(\oz) = \nu(C_1)-\nu(C_2)+\nu(C_3)
\end{equation}
{Next, note that, for $j\more2$,
\begin{align}\label{eqn:rel:nuCj} \nonumber
\nu(C_{j-1})&= \int_{\Lambda}P^{ext}[\oz,C_{j-1}]\nu(d\bar{z})\\ \nonumber
%				&= \sum\limits_{\oz\in C_{j-2}} P^{ext}[\oz,C_{j-1}]\nu(\oz) + \sum\limits_{\oz\in C_{j}} P^{ext}[\oz,C_{j-1}]\nu(\oz)\\ \nonumber
%				&= \sum\limits_{\oz\in C_{j-2}} P^{ext}[\oz,C_{j-1}]\nu(\oz) + \sum\limits_{\oz\in C_{j}}(1-P^{ext}[\oz,C_{j+1}])\nu(\oz)\\ \nonumber
				&= \sum\limits_{\oz\in C_{j-2}} P^{ext}[\oz,C_{j-1}]\nu(\oz) + \sum\limits_{\oz\in C_{j}}(1-P^{ext}[\oz,C_{j+1}])\nu(\oz)\\
				&= \nu(C_{j}) + \sum\limits_{\oz\in C_{j-2}} P^{ext}[\oz,C_{j-1}]\nu(\oz) -  \sum\limits_{\oz\in C_{j}}P^{ext}[\oz,C_{j+1}]\nu(\oz).
\end{align}
Therefore,}
\begin{equation}\label{eqn:rel:nuCj2}
\sum\limits_{\oz\in C_{j}}P^{ext}[\oz,C_{j+1}]\nu(\oz) =\nu(C_{j}) - \nu(C_{j-1}) + \sum\limits_{\oz\in C_{j-2}} P^{ext}[\oz,C_{j-1}]\nu(\oz)
\end{equation}
Applying {the previous equation recursively, for all even $j\geq2$,
\begin{align*}
\sum\limits_{\oz\in C_{j}}P^{ext}[\oz,C_{j+1}]\nu(\oz) &=\nu(C_{j}) - \nu(C_{j-1}) + \nu(C_{j-2})-\nu(C_{j-3})+\ldots+\nu(C_4)-\nu(C_3)+\sum\limits_{\oz\in C_{2}} P^{ext}[\oz,C_3]\nu(\oz),
\end{align*} 
and, thus, by (\ref{eqn:rel:nuC3}), 
\begin{equation}\label{eqn:rel:nu:sums:even}
\sum\limits_{\oz\in C_{j}}P^{ext}[\oz,C_{j+1}]\nu(\oz) =\nu(C_{j}) - \nu(C_{j-1}) + \nu(C_{j-2})-\nu(C_{j-3})+\ldots+\nu(C_4)-\nu(C_3)+\nu(C_{2})-\nu(C_{1}).
\end{equation} 
However, if $j\geq{}$ is odd, by (\ref{eqn:rel:nuC4}),
\begin{align}\nonumber
\sum\limits_{\oz\in C_{j}}P^{ext}[\oz,C_{j+1}]\nu(\oz) &=\nu(C_{j}) - \nu(C_{j-1}) + \nu(C_{j-2})-\nu(C_{j-3})+\ldots+\nu(C_5)-\nu(C_4)+\sum\limits_{\oz\in C_{3}} P^{ext}[\oz,C_4]\nu(\oz)\\
\label{eqn:rel:nu:sums}
&=\nu(C_{j}) - \nu(C_{j-1}) + \nu(C_{j-2})-\nu(C_{j-3})+\ldots+\nu(C_5)-\nu(C_4)+ \nu(C_1)-\nu(C_2)+\nu(C_3).
\end{align} 
%\begin{equation}\label{eqn:rel:nu:sums}
%\sum\limits_{\oz\in C_{j}}P^{ext}[\oz,C_{j+1}]\nu(\oz) =\nu(C_{j}) - \nu(C_{j-1}) + \nu(C_{j-2})-\nu(C_{j-3})+\ldots+\nu(C_4)-\nu(C_3)+\nu(C_2)-\nu(C_1).
%\end{equation} 
Equations (\ref{eqn:rel:nu:sums:even})-(\ref{eqn:rel:nu:sums})} imply {that
%\[
%\sum\limits_{\oz\in C_{j}}P^{ext}[\oz,C_{j+1}]\nu(\oz) + \sum\limits_{\oz\in C_{j+1}}P^{ext}[\oz,C_{j+2}]\nu(\oz)= \nu(C_{j+1}),
%\]
%and, therefore,
\[
\nu(C_{j+1}) = \sum\limits_{\oz\in C_{j}}P^{ext}[\oz,C_{j+1}]\nu(\oz) + \sum\limits_{\oz\in C_{j+1}}P^{ext}[\oz,C_{j+2}]\nu(\oz).
\]
However,} by the definition of a stationary measure,
\[
\nu(C_{j+1}) = \sum\limits_{\oz\in C_{j}}P^{ext}[\oz,C_{j+1}]\nu(\oz) + \sum\limits_{\oz\in C_{j+2}}P^{ext}[\oz,C_{j+1}]\nu(\oz).
\]
The previous two equations {yield the following} relation,\footnote{{Eq.~(\ref{eqn:Detailed:Balanced}) gives some insight into the structure of the stationary measure $\nu$ and can be regarded as a  ``batch'' version of the so-called Detailed Balance Conditions for Markov Chains, which are important for analyzing reversible processes (c.f. \cite{Kelly11}[Chapter 1.2]). }}
\begin{equation}\label{eqn:Detailed:Balanced}
\sum\limits_{\oz\in C_{j+1}}P^{ext}[\oz,C_{j+2}]\nu(\oz)= \sum\limits_{\oz\in C_{j+2}}P^{ext}[\oz,C_{j+1}]\nu(\oz).
\end{equation}
{We are now ready to bound the term $\nu(C_j)$. To that end, notice} that for any $\oz\in C_j$, $P^{ext}[\oz,C_{j-1}]\geq (1-p_{\mathbf{1}})$ {and, thus,}
\begin{align*}
\nu(C_j)
%&=\int_{\Lambda}P^{ext}[\oz,C_{j}]\nu(d\bar{z})\\ 
				&= \sum\limits_{\oz\in C_{j-1}} P^{ext}[\oz,C_{j}]\nu(\oz) + \sum\limits_{\oz\in C_{j+1}} P^{ext}[\oz,C_{j}]\nu(\oz)\\
				&= \sum\limits_{\oz\in C_{j}} P^{ext}[\oz,C_{j-1}]\nu(\oz) + \sum\limits_{\oz\in C_{j+1}} P^{ext}[\oz,C_{j}]\nu(\oz)\\
				&\geq (1-p_{\mathbf{1}}) \nu(C_j) + (1-p_{\mathbf{1}})\nu(C_{j+1}).
\end{align*}
{Therefore, $\nu(C_{j+1}) \leq \nu(C_j)p_{\mathbf{1}}/(1-p_{\mathbf{1}})$,} 
which, {by induction,}  implies that
\begin{equation}\label{eqn:nus:geom}
\nu(C_{j+1}) \leq \left(\frac{p_{\mathbf{1}}}{1-p_{\mathbf{1}}}\right)^{j-1}\nu(C_2),
\end{equation}
Finally, by (\ref{eqn:compar:sigma}) {and} (\ref{eqn:nus:geom}) and the fact that $p^* =p_{\mathbf{1}}/ (1-p_{\mathbf{1}}) \less1$, {for some bounded constant $C^{*}$,}
\begin{align}\label{eqn:bddsigmaf} \nonumber
\nu(\Ex^2_{\cdot}(\sigma_{C_1}))
%&=\sum_{\oz\in\Lambda}\nu(\oz)\Ex_{\oz}^2(\sigma_{C_1})\\ \nonumber
		&=\sum_{\oz\in C_1}\nu(\oz)\Ex_{\oz}^2(\sigma_{C_1})+ \sum\limits_{j=2}^\infty \sum_{\oz\in C_j}\nu(\oz)\Ex_{\oz}^2(\sigma_{C_1})\\ \nonumber
%		&\leq C^*+ \sum\limits_{j=2}^\infty \sum_{\oz\in C_j}\nu(\oz)\left(\frac{j-1}{1-2p_{\mathbf{N^*}}}\right)^2\\ \nonumber
		&= C^* + \sum\limits_{j=2}^\infty \nu(C_j)\left(\frac{j-1}{1-2p_{\mathbf{N^*}}}\right)^2\\
		&\leq C^* + \nu(C_1)\sum\limits_{j=2}^\infty \left(p^*\right)^{j-1} \left(\frac{j-1}{1-2p_{\mathbf{N^*}}}\right)^2\less\infty.
\end{align}
{It only remains to} show that the variance in the FCLT can be written as (\ref{exprsigma}). But, by Theorem 17.5.3 {in \cite{Meyn}}, it is enough to  show that he Markov chain $\{V_n\}_{n\geq1}$ is ergodic and there exists a function $F:\Lambda\to[0,\infty]$ {such that $\nu(F^2)\less\infty$ and}
\begin{equation}\label{eqn:cond:V3}
\Delta F(\oz) \leq - 1 +b\indicator{\oz \in B}, 
\end{equation}
for {a constant $b\less\infty$ and a finite set $B$, where $\Delta$ is the operator $\Delta F(\oz):=\Ex_{\oz}\left[F(V_1)-F(V_0)\right]$} (c.f. Section 14.2.1 in \cite{Meyn}). However, it is known (c.f. Section 13.1.2 \cite{Meyn}) that aperiodic positive Harris chains over a countable state space are ergodic and by {Proposition 14.1.2 and Theorem 14.2.3(ii) therein}, the function ${{F}(\oz)}:=\Ex_{{\oz}}[\sigma_B]$ satisfies (\ref{eqn:cond:V3}), where $\sigma_B$ is the first hitting time of the set $B$ in (\ref{eqn:cond:V3}). Since it was proven above that $\nu(\Ex_{\cdot}^2[\sigma_{C_1}])\less\infty$, {we take $B=C_1$ to conclude the proof}.
\end{proof}

\begin{proof}[{\textbf{Proof of Lemma \ref{jointdist}}}] 
{Let us start by noting that the generator of the two-dimensional random walk $Y$ is given by the finite difference operator $\mathscr{L}$ defined in (\ref{uGenerator}). More concretely, for a function $\phi:\bar\Omega_{N^{*}}\to\mathbb{R}$, $\mathscr{L}\phi(x,y)$ is defined analogously to (\ref{uGenerator}) but replacing $u(t,x,y)$ with $\phi(x,y)$.}
For simplicity, we denote {$\varsigma:=\varsigma(x,y)$ and remark that $\varsigma$} is an absolutely continuous random variable. {Let $\bar{u}(t,x,y)$ be an arbitrary bounded function such that $t\mapsto \bar{u}(t,x,y)$ is $C_1$ for all $(x,y)$ and $(t,x,y)\mapsto \partial_{t}\bar{u}(t,x,y)$ is bounded. Fix $T\more0$ and let $f(t,x,y)=\bar{u}(T-t,x,y)$. Under the stated conditions,} $\bar{u}$ belongs to the domain of the generator $\mathscr{L}$ and, thus, the process
\[
f(t,X_t)-\int_0^t\left(\frac{\partial}{\partial r}+\mathscr{L}\right)f(r,{Y_r})dr,\quad t\in[0,T]
\]
is a local martingale. Therefore,
\[
M_t:=\bar{u}(T-t,{Y_t})-\int_0^t\left(-\frac{\partial}{\partial {t}}+\mathscr{L}\right)\bar{u}(T-r,{Y_r})dr,\quad {t\in [0,T]},
\]
is a martingale. Let $\sigma:=T\wedge{\varsigma}$. {By the Optional Sampling Theorem},
\begin{equation}\label{AppendixC:eq1}
\bar{u}(T,x,y)=\Ex[\bar{u}(T-\sigma,{Y_\sigma})]-\Ex\left[\int_0^\sigma\left(-\frac{\partial}{\partial {t}}+\mathscr{L}\right)\bar{u}(T-r,{Y_r})dr\right].
\end{equation}
{Now, suppose that $\bar{u}(t,x,y)$ solves the following} initial value problem,
\begin{equation}\label{AppendixC:PDE}
\left\{\begin{array}{rcl}
\left(-\frac{\partial}{\partial {t}} + \mathscr{L}\right) \bar{u} (T-r,x,y)=0 & \text{for} & 0\leq r \leq T,\;  
(x,y)\in\{1,2,\ldots,M\}^2.\\
\bar{u}(T-r,x,y)=\indicator{(x,y)=\ax} & \text{for} & 0\leq r\leq T,\, (x,y)\in\mathscr{A}.\\
\bar{u}(0,x,y)=\indicator{(x,y)=\ax} & \text{for} & (x,y)\in\{0,1,2,\ldots,M\}^2.\end{array}\right.
\end{equation}
In that case, by Eq. (\ref{AppendixC:eq1}),
\begin{align*}
		\bar{u}(T,x,y)&=\Ex[\bar{u}(T-\sigma,{Y}_\sigma)]\\
				&=\Ex[\bar{u}(T-\sigma,{Y}_\sigma)\indicator{\sigma\less T}+\bar{u}(T-\sigma,{Y}_\sigma)\indicator{\sigma= T}]\\
				&=\Ex[\indicator{Y_\sigma=\ax}\indicator{\sigma<T} + \bar{u}(0,Y_\sigma)\indicator{\sigma= T}]\\
				&=\Px[Y_{\varsigma}=\ax, \varsigma\leq T],
\end{align*}
which implies that $\bar{u}(t,x,y)=u_{{\ax}}(t,x,y)$.
\end{proof}

\begin{proof}[{\textbf{Proof of Proposition \ref{lemmausol}}}] 
Let 
\[
	w(t,x,y)=\left(\frac{\lambda}{\up}\right)^{\frac{\ax_1+\ax_2}{2}}e^{t(2(\lambda+\up)-
	4\sqrt{\lambda\up})}\indicator{(x,y)=\ax}.
\] 
Fix  $\ov(t,x,y)=v(t,x,y)-w(t,x,y)$ and note that $v(t,x,y)$ satisfies the system (\ref{PDEv}) if and only if $\ov$ is a solution to the initial value problem:
\begin{equation}\label{PDEov}
\left\{\begin{array}{rlcl}
-\left(\frac{\partial}{\partial t} - \sqrt{\lambda\up}\Delta\right) \ov(t,x,y)&=\left(\frac{\partial}{\partial t}-\sqrt{\lambda\up}\Delta\right)w(t,x,y) & \text{for} & t\geq0,  
\quad (x,y)\in\Omega_{N^{*}}\\
\ov(t,x,y)&=0 & \text{for} & t\geq0,\quad (x,y)\in\mathscr{A}\\
\ov(0,x,y)&=0 & \text{for} & (x,y)\in{\bar\Omega}_{N^{*}}.\end{array}\right.
\end{equation}
Let 
$\{\psi_{k}(t)\}_{k=1}^{{N^{*}}^2}$ and $\{\varsigma_{k}^{\bar{a}}\}_{k=1}^{{N^{*}}^2}$ be such that $\ov(t,x,y)=\sum_{k}^{{N^{*}}^{2}}\psi_k(t)f_k(x,y)$ and $\indicator{(x,y)=\overline{a+1}}=\sum_{k=1}^{{N^{*}}^2}\varsigma_k^{\ax}f_k(x,y)$. 
Using the first representation, the left-hand side of the first equation in (\ref{PDEov}) becomes 
\begin{align*}
-\left(\frac{\partial}{\partial t} - \sqrt{\lambda\up}\Delta\right) \ov(t,x,y)
		=-\sum\limits_{k=1}^{\ {N^{*}}^2} \psi_k'(t)f_k(x,y) - \psi_k(t)\sqrt{\lambda\up}\Delta f_k(x,y)=\sum\limits_{k=1}^{\ {N^{*}}^2} \left[-\psi_k'(t) + \psi_k(t)\sqrt{\lambda\up}\xi_k\right]f_k(x,y).
\end{align*}
Similarly, 
 the right-hand side of the first equation in (\ref{PDEov}) is given by
\begin{align*}
\left(\frac{\partial w}{\partial t}-\sqrt{\lambda\up}\Delta\right) w
		=-\left(\frac{\lambda}{\up}\right)^{\frac{\ax_1+\ax_2}{2}} \sqrt{\lambda\up}e^{t(2(\lambda+\up)-4\sqrt{\lambda\up})}\indicator{(x,y)=\overline{a+1}}=-e^{t(2(\lambda+\up)-4\sqrt{\lambda\up})}\left(\frac{\lambda}{\up}\right)^{\frac{\ax_1+\ax_2}{2}} \sqrt{\lambda\up}\sum\limits_{k=1}^{{N^{*}}^2}\varsigma_k^{\ax}f_k(x,y).
\end{align*}
Combining the previous two expressions and recalling that $\{f_k(x,y)\}_{k}$ is an orthonormal basis, {it follows that the function $\ov(t,x,y)=\sum_{k}\psi_k(t)f_k(x,y)$ will solve the system (\ref{PDEov}) if and only if, for every $k$, the function $\psi_{k}(t)$ satisfies the following equation:
\begin{equation} \label{ODEov}
-\psi_k'(t) + \psi_k(t)\sqrt{\lambda\up}\xi_k=-e^{t(2(\lambda+\up)-4\sqrt{\lambda\up})}\left(\frac{\lambda}{\up}\right)^{\frac{\ax_1+\ax_2}{2}} \sqrt{\lambda\up}\varsigma_k^{\ax},
\end{equation}
{with the initial condition $\psi_{k}(0)=0$. It is easy to see that the previous differential equation is well posed and has solution}
\begin{equation}\label{solODEov}\nonumber
\psi_k(t)=\left(\frac{\lambda}{\up}\right)^{\frac{\ax_1+\ax_2}{2}}\frac{\sqrt{\lambda\up}\varsigma_k^{\ax}}{2(\lambda+\up)-\sqrt{\lambda\up}(4+\xi_k)}\left[e^{t\sqrt{\lambda\up}\xi_k}-e^{t(2(\lambda+\up)-4\sqrt{\lambda\up})}\right].
\end{equation}
Therefore, 
\[
\ov(t,x,y)=\left(\frac{\lambda}{\up}\right)^{\frac{\ax_1+\ax_2}{2}}\sum\limits_{k}\frac{\sqrt{\lambda\up}\varsigma_k^{\ax}}{2(\lambda+\up)-\sqrt{\lambda\up}(4+\xi_k)}\left[e^{t\sqrt{\lambda\up}\xi_k}-e^{t(2(\lambda+\up)-4\sqrt{\lambda\up})}\right]f_k(x,y),
\]
satisfies the initial value problem (\ref{PDEov}), which in turn, implies that 
\[
u(t,x,y)=\left(\frac{\lambda}{\up}\right)^{\frac{\ax_1+\ax_2-x-y}{2}}\left[\sum\limits_{k=1}^{\ {N^{*}}^2}\frac{\sqrt{\lambda\up}\varsigma_k^{\ax}}{2(\lambda+\up)-\sqrt{\lambda\up}(4+\xi_k)}\left(1 - e^{-t(2(\lambda+\up)-(4+\xi_k)\sqrt{\lambda\up})}\right)f_k(x,y)+\indicator{(x,y)=\ax}\right],
\]
is a solution of (\ref{PDEu}). Then, the representation (\ref{uexplicit})  immediately follows by noting that $\varsigma_k^{\ax}=f_{k}(\overline{a+1})$ and rewriting the previous expression in terms of $\chi=\lambda/\upsilon$. } 
\end{proof}

\begin{lemma} \label{vproblem} {A function $u:[0,T]\times\bar\Omega_{{N^{*}}}\to\mathbb{R}$ is a solution of the system of differential equations (\ref{PDEu}) if and only if the function $v(t,x,y)$ defined by 
\[
v(t,x,y)=\left(\frac{\lambda}{\up}\right)^{\frac{x+y}{2}}e^{t(2(\lambda+\up)-4\sqrt{\lambda\up})}u(t,x,y),
\]
solves the system of difference equations
\begin{equation}\label{PDEv}
\left\{\begin{array}{ll}
\left(-\frac{\partial}{\partial t} + \sqrt{\lambda\up}\Delta\right) v(t,x,y)=0, & \text{for} \quad t\geq0,  
\quad (x,y)\in\Omega_{N^{*}},\\
v(t,x,y)=\left(\frac{\lambda}{\up}\right)^{\frac{\ax_1+\ax_2}{2}}e^{t(2(\lambda+\up)-4\sqrt{\lambda\up})}\indicator{(x,y)=\ax}, & \text{for} \quad t\geq0,\quad (x,y)\in\mathscr{A},\\
v(0,x,y)=\left(\frac{\lambda}{\up}\right)^{\frac{\ax_1+\ax_2}{2}}\indicator{(x,y)=\ax}, & \text{for} \quad (x,y)\in{\bar\Omega}_{N^{*}},\end{array}\right.
\end{equation}
where $\ax=(\ax_1,\ax_2)\in\mathscr{A}$ and, for each fixed $t$, $\Delta v(t,x,y)$ is defined as in (\ref{vGenerator}) with respect to $x$ and $y$}.
\end{lemma}

\begin{proof}%[{\textbf{Proof of Lemma \ref{vproblem}}}] 
The proof is standard and is omitted. 
\end{proof}

\end{document}